%% file: main.tex
\newif\ifDRAFT 
\DRAFTfalse

\documentclass[11pt]{article}
\usepackage[margin=1in]{geometry}
\usepackage{graphicx} %
\usepackage[framemethod=TikZ]{mdframed}
\usepackage{amsfonts, amsmath, amssymb, amsthm}
\allowdisplaybreaks
\usepackage{mathrsfs}  
\usepackage{algorithm}
\usepackage[noend]{algpseudocode}
\usepackage{hyperref}
\usepackage{wrapfig}
\usepackage{color}
\usepackage{array}
\usepackage{cleveref}
\definecolor{DarkRed}{rgb}{0.5,0.1,0.1}
\definecolor{DarkBlue}{rgb}{0.1,0.1,0.5}
\definecolor{ForestGreen}{rgb}{0.1333,0.5451,0.1333}
\hypersetup{colorlinks=true, linkcolor=DarkRed,citecolor=ForestGreen}
\usepackage[inkscapeformat=pdf]{svg}
\usepackage{multirow}
\usepackage{mathtools}
\usepackage{enumitem}
\usepackage{thm-restate}
\usepackage{tikz}

\usepackage{tabulary} %
\usepackage{booktabs}
\usepackage{arydshln}

\ifDRAFT
    \usepackage{lineno}
    \linenumbers
\fi

\theoremstyle{plain}
\newtheorem{theorem}{Theorem}[section]
\newtheorem{lemma}[theorem]{Lemma}
\newtheorem{corollary}[theorem]{Corollary}
\newtheorem{observation}[theorem]{Observation}
\newtheorem{claim}[theorem]{Claim}

\newtheorem{conjecture}{Conjecture}[section]

\theoremstyle{definition}
\newtheorem{definition}[theorem]{Definition}

\newtheorem{remark}[theorem]{Remark}
\newtheorem{problem}[conjecture]{Problem}

\crefname{equation}{Eqn.}{Eqns.}

\DeclareSymbolFont{bbold}{U}{bbold}{m}{n}
\DeclareSymbolFontAlphabet{\mathbbold}{bbold}

\newcommand{\IGNORE}[1]{}

\newcommand{\T}{\mathcal{T}}
\newcommand{\D}{\mathcal{D}}

\newcommand{\poly}{\operatorname{poly}}

\newcommand{\TEOracle}{\mathsf{TEcutdet}}
\newcommand{\USOracle}{\mathsf{UScutdet}}
\newcommand{\FEOracle}{\mathsf{FewTcutdet}}

\newcommand{\cO}{\mathcal{O}}
\newcommand{\cF}{\mathcal{F}}
\newcommand{\sC}{\mathscr{C}}
\usepackage[normalem]{ulem} %

\mdfdefinestyle{MyFrame}{%
    roundcorner=10pt,
    innertopmargin=\baselineskip,
    innerbottommargin=\baselineskip,
    innerrightmargin=10pt,
    innerleftmargin=10pt,
    backgroundcolor=gray!5!white}

\title{New Oracles and Labeling Schemes for Vertex Cut Queries}

\ifDRAFT
    \author{Anonymous Authors}
\else
    \author{
    Yonggang Jiang\thanks{MPI-INF and Saarland University, Germany. Email: \texttt{yjiang@mpi-inf.mpg.de}.}
    \and
    Merav Parter\thanks{Weizmann Institute, Israel. Email: \texttt{merav.parter@weizmann.ac.il}. Supported by the European Research Council (ERC) under the European Union’s Horizon 2020 research and
    innovation programme, grant agreement No. 949083.}
    \and
    Asaf Petruschka\thanks{Weizmann Institute. Email: \texttt{asaf.petruschka@weizmann.ac.il}.  Supported by an Azrieli Foundation fellowship.}
    }
\fi
\date{}

\begin{document}

\maketitle
\pagenumbering{roman}
\input{abstract}
\newpage

\tableofcontents
\newpage
\pagenumbering{arabic}

\input{intro}
\input{prelim}

\input{cut-labels}
\section{Vertex Cut Oracles}\label{sec:oracles}
\input{overview-oracles}
\input{special-cut-detectors}

\input{left-right-graphs}

\input{LRtree}
\input{data-structure}
\input{f-connected-DS}
\input{space-improvement}
\input{cut-respecting-family}

\input{space-lower-bound}
\input{acknowledgments}

\bibliographystyle{alphaurl}
\bibliography{references}

\begin{appendix}
    \input{lower-bounds}
\end{appendix}

\end{document}

%% file: abstract.tex
\begin{abstract}
We study the succinct representations of vertex cuts by centralized oracles and labeling schemes.
For an undirected $n$-vertex graph $G = (V,E)$ and integer parameter $f \geq 1$, the goal is supporting vertex cut queries: Given $F \subseteq V$ with $|F| \leq f$, determine if $F$ is a vertex cut in $G$.
In the centralized data structure setting, it is required to preprocess $G$ into an \emph{$f$-vertex cut oracle} that can answer such queries quickly, while occupying only small space. 
In the labeling setting, one should assign a \emph{short label} to each vertex in $G$, so that a cut query $F$ can be answered by merely inspecting the labels assigned to the vertices in $F$.

While the ``$st$ cut variants'' of the above problems have been extensively studied and are known to admit very efficient solutions, the basic (global) ``cut query'' setting is essentially open (particularly for $f > 3$).
This work achieves the first significant progress on these problems:

\begin{itemize}
\item \textbf{$f$-Vertex Cut Labels:}
Every $n$-vertex graph admits an $f$-vertex cut labeling scheme, where the labels have length of $\tilde{O}(n^{1-1/f})$ bits (when $f$ is polylogarithmic in $n$).
This nearly matches the recent lower bound given by Long, Pettie and Saranurak (SODA 2025).

\item \textbf{$f$-Vertex Cut Oracles:}
For $f=O(\log n)$, every $n$-vertex graph $G$ admits $f$-vertex cut oracle with $\tilde{O}(n)$ space and $\tilde{O}(2^f)$ query time.
We also show that our $f$-vertex cut oracles for every $1 \leq f \leq n$ are optimal up to $n^{o(1)}$ factors (conditioned on plausible fine-grained complexity conjectures).
If $G$ is $f$-connected, 
i.e., when one is interested in \emph{minimum} vertex cut queries, 
the query time improves to $\tilde{O}(f^2)$, for any $1 \leq f \leq n$.
\end{itemize}

Our $f$-vertex cut oracles are based on a special form of hierarchical expander decomposition that satisfies some ``cut respecting'' properties. Informally, we show that any $n$-vertex graph $G$ can be decomposed into terminal vertex expander graphs that ``capture'' all cuts of size at most $f$ in $G$. The total number of vertices in this graph collection is $\tilde{O}(n)$. 
We are hopeful that this decomposition will have further applications (e.g., to the dynamic setting).

\end{abstract}

%% file: intro.tex
\section{Introduction }

The study of \emph{vertex cuts} in graphs has a long and rich history, dating back to works by Menger (1927, \cite{menger1927allgemeinen}) and Whitney (1932, \cite{whitney1932congruent}). 
In recent years, substantial progress has been made in understanding vertex cuts from computational standpoints
(see e.g.,~\cite{Censor-HillelGK14,NanongkaiSY19,LiNPSY21,SaranurakY22,PettieSY22,HuangLSW23,JiangM23,BJMY2024,JNDSY2025} and references therein),
helping to bridge long-standing gaps with edge cuts.
A vast majority of this work has been focused on computing the vertex connectivity, i.e., the minimum size of a vertex cut.
Despite exciting advancements, (global) vertex cuts continue to pose intriguing open algorithmic challenges, being a topic of significant interest and activity.
In this paper, we study vertex cuts from a \emph{data structures} perspective, examining centralized \emph{oracles} and distributed \emph{labeling schemes};
both these data structures fit the following description:
\begin{problem}[$f$-Vertex Cut Data Structure]\label{problem:f-vertex-cut-ds}
Given an undirected graph $G = (V,E)$ and integer parameter $f \geq 1$, the goal is to preprocess it into a data structure that answers ``is-it-a-cut'' queries: on query $F \subseteq V$ s.t.\ $|F| \leq f$, decide if $F$ is a vertex cut in $G$ (i.e., if $G-F$ is disconnected). 
\end{problem}

In the centralized oracle setting, we measure the \emph{space}, \emph{query time} and \emph{preprocessing time} complexities.
In the labeling scheme setting, the preprocessing should produce a succinct label $L(v)$ for every vertex $v \in V$, and the query algorithm is restricted to merely using the information stored in the labels of the query vertices $F$.
The most important complexity measure is the (maximum) \emph{label length}
$\max_{v \in V} |L(v)|$, where $|L(v)|$ denotes the bit-length of $L(v)$.

The ``$st$ variants'' of the above data structures, commonly referred to as $f$-\emph{vertex failure connectivity oracles/labels}, have been introduced by Duan and Pettie \cite{DuanP10}  and extensively studied in recent years~\cite{DuanP20,HenzingerN16,BrandS19,PilipczukSSTV22,kosinas:LIPIcs.ESA.2023.75,HuK024,LongW24,ParterP22a,ParterPP24,LongPS25}.
In the $st$ variant, the query consists of $F$ along with two vertices $s,t \in V$, and should determine if $s,t$ are connected in $G-F$.
(This is sometimes split further into an \emph{update} phase where only $F$ is revealed, followed by connectivity queries in $G-F$ of different vertex pairs $s,t$.)
The current state-of-the-art oracles have converged into deterministic $O(m) + (fn)^{1+o(1)} + 
\tilde{O}(f^2 n)$ preprocessing time, $\widetilde{O}(f n)$ space, $\widetilde{O}(f^2)$ update time and $\widetilde{O}(f)$ query time, by Long and Wang~\cite{LongW24}.%
        \footnote{
        Throughout, the $\tilde{O}(\cdot)$ notation hides $\poly \log n$ factors.
        }
All these complexities are almost optimal due to (conditional) lower bounds by Long and Saranurak~\cite{LongS22}.
As for labels, the current record is label length of $\tilde{O}(f^2)$ bits (randomized) or $\tilde{O}(f^4)$ bits (deterministic) by Long, Pettie and Saranurak~\cite{LongPS25}, where the known lower bound is  $\Omega(f + \log n)$~\cite{ParterPP24}.

However, much less is known on global vertex cut data structures.
To this date, efficient $f$-vertex cut oracles are known only for $f\leq 3$~\cite{BattistaT89,CohenBKT93}.
Pettie and Yin~\cite{PettieY21} posed the question of designing efficient $f$-vertex cut oracles for any $f$, even under the simplified assumption that the given graph is $f$-vertex connected.
Long and Saranurak~\cite{LongS22} have shown that rather surprisingly, global $f$-vertex cut oracles must have much higher complexities then their $st$ variants: even when $f=n^{o(1)}$, the query time is at least $n^{1-o(1)}$ (with $\poly(n)$ preprocessing time).
Such phenomena appear also with labeling schemes: Parter, Petruschka and Pettie~\cite{ParterPP24} posed the question of designing global vertex cut labels and showed their length must be at least exponential in $f$, later strengthened to $\Omega(n^{1-1/f} /f)$ in~\cite{LongPS25}.
Providing non-trivial upper bounds for vertex cut labels is open for any $f \geq 2$.

\paragraph{Our Contribution.}
This work provides the first significant progress on global $f$-vertex cut data structures (\Cref{problem:f-vertex-cut-ds}), both centralized oracles and labeling schemes, obtaining almost-optimal bounds (possibly conditional on popular conjectures).
In light of the above-mentioned progress on the $st$-cut variants, our data structures are all generally based on reducing a single global cut query into a bounded number of $st$-cut queries.
This strategy is inspired by the recent exciting line of works and breakthrough results in the \emph{computational} setting, that are all based on this global to $st$ reduction scheme. A prominent example is the randomized algorithm by Li et al.~\cite{LiNPSY21} for computing (global) vertex connectivity in polylogarithmic max-flows; some other instances are~\cite{Cen0NPSQ21,NSYArxiv23,HeHS24}.
We now turn to discuss our results in depth.

\paragraph{Minimum Vertex Cut Oracles.}
We address the open problem posed by Pettie and Yin \cite{PettieY21} of answering queries of the following form: Is $F$ a \emph{minimum} vertex cut in $G$?
This can be viewed as a $f$-vertex cut oracles for $f$-vertex connected graphs.
We show:

\begin{theorem}[Minimum Vertex Cut Oracles] \label{thm:cut-theorem-fconnected}
    Let $G = (V,E)$ be a graph with $n$ vertices and $m$ edges.
    Let $f \geq 1$ be such that $G$ is $f$-vertex connected.
    Then, there is a deterministic minimum vertex cut oracles for $G$ with:
    \begin{itemize}
        \item $\tilde{O}(fn)$ space,
        \item $\tilde{O}(f^2)$ query time, and
        \item $O(m) + \tilde{O}(f^2 n) + \tilde{O}((fn)^{1+\delta})$ preprocessing time (for any constant $\delta > 0$).
    \end{itemize}
    The last term in the preprocessing time can be further improved to $fn^{1+o(1)}$, at the cost of increasing the space and query time by $n^{o(1)}$ factors.
\end{theorem}
Our space bound improves over the $O(f^2 n \log n)$ bound obtained by the subsequent work of~\cite{Kosinas25}.
In fact, by~\Cref{thm:spaceLB} which is discussed later, the space bound is optimal up to polylogartimic factors, irrespectively of query or preprocessing time.
Our query time also improves over the $O(f^4 \log n)$ bound of~\cite{Kosinas25}. 
We note that $O(f^2)$ is a natural barrier for our approach, which reduces a (global) cut query $F$ into a bounded number of $st$-connectivity queries in $G-F$.
In the latter setting, this query (or more precisely, update) time is known to be conditionally tight~\cite{LongS22} (i.e., already for a \emph{single} $st$-cut query); thus, an improvement seems to call for an entirely different approach. %
\paragraph{General Vertex Cut Oracles.}
Next, we address the problem of oracles that answer general $f$-vertex cut queries, where the graph is not necessarily $f$-connected.
We show: 
\begin{theorem}[$f$-Vertex Cut Oracles]\label{thm:main-cut-theorem}
    Let $G = (V,E)$ be a graph with $n$ vertices and $m$ edges.
    Let $f = O(\log n)$ be a nonnegative integer.
    There is a deterministic $f$-vertex cut oracle for $G$ with:
    \begin{itemize}
        \item $\tilde{O}(n)$ space,
        \item $\tilde{O}(2^{|F|})$ query time (where $F \subseteq V$ s.t.\ $|F| \leq f$ is the given query), and
        \item $O(m) + \tilde{O}(n^{1+\delta})$ preprocessing time (for any constant $\delta > 0$).%
        \footnote{
        Decreasing $\delta$ leads to larger $\poly \log n$ factors hidden in the $\tilde{O}(\cdot)$ notations.}
    \end{itemize}
    The last term in the preprocessing time can be further improved to $n^{1+o(1)}$, at the cost of increasing the space and query time by $n^{o(1)}$ factors.  
\end{theorem}
Note that there is also a trivial $f$-vertex cut oracle that simply applies the Nagamochi-Ibaraki sparsification~\cite{NagamochiI92} and stores the resulting subgraph using $\tilde{O}(\min\{m,fn\})$ space and query time, and $O(m)$ preprocessing time.
As it turns out, combining~\Cref{thm:main-cut-theorem} with the above trivial approach (i.e., using the former when $f \leq \log n$ and the latter when $f \geq \log n$) gives $f$-vertex cut oracles that are \emph{almost-optimal} (up to $n^{o(1)}$ factors), for every value of $f$.
The almost-optimally of the preprocessing time is clear, so we focus only on space and query time in the following discussion.

\begin{description}
    \item[Space:] 
    We give an information-theoretic lower bound on the space occupied by the oracle, which is irrespective of preprocessing or query time.
    This lower bound holds also for the special case of minimum vertex cut oracles (i.e., where the graph is $f$-connected):
    \begin{theorem}\label{thm:spaceLB}
    Let $n, f$ be positive integers such that $2 \leq f \leq n/4$.
    Any $f$-vertex cut connectivity oracle   must occupy $\Omega \big( fn \cdot \log\big(\frac{n}{f}\big) \big)$ bits on some $f$-connected $n$-vertex graph.%
    \footnote{
        If $f > \frac{n}{4}$, then a lower bound of $\Omega(n^2)$ still holds, since an $f$-vertex cut oracle is also an $f'$-vertex cut oracle for any $f' \leq f$ by definition.
        But in this case, the lower bound instance is not $f$-connected but only $\Omega(f)$-connected.
        }$^,$\footnote{
        \Cref{thm:spaceLB} does not contradict the sparsification of~\cite{NagamochiI92}, implying that $G$ has a subgraph of $O(fn)$ \emph{edges} with exactly the same vertex cuts of size $\leq f$.
        By similar counting arguments as in the proof of~\Cref{thm:spaceLB}, one can show that $\Omega(fn \log(f/n))$ \emph{bits} are required to represent an arbitrary graph with $O(fn)$ \emph{edges} in the worst case.
        }
    For $f=1$, the lower bound is $\Omega(n)$ bits.
    \end{theorem}
    As observed in~\cite{DuanP20}, in the $st$ variant a space lower bound of $\Omega(f n)$ bits is easy to show by considering the subgraphs of the complete bipartite graph $K_{n,f+1}$. However, this does not yield a lower bound for the global variant.
    We therefore provide a different construction which yields a slightly stronger bound due to the extra $\log(n/f)$ factor. We note that our lower bound also holds for the $st$ variant of ~\cite{DuanP20}. (This holds as one can make $n$ queries to the $st$-cut oracle of ~\cite{DuanP20} to determine if the queried set of vertices $F$ is a cut.) Therefore the bound of \Cref{thm:spaceLB} also strengthens the known bound of $st$-cut oracles.

    \item[Query time:]
    We first observe that lower bound construction of Long and Saranurak~\cite{LongS22} essentially also implies that the dependency on $f$ in the query time of~\Cref{thm:main-cut-theorem} cannot be made polynomial.
    Specifically, already for $f = c \log n$ (where $c$ is an absolute constant), any oracle with preprocessing time $n^{2-o(1)}$ must have query time $n^{1-o(1)}$ unless the Strong Exponential Time Hypothesis fails.\footnote{
    We note that the lower bound instance is a sparse graph with only $O(fn) = \tilde{O}(n)$ edges, so the oracle of~\Cref{thm:main-cut-theorem} preprocesses it with much less time than $n^{2-o(1)}$.
    }
    Thus, for $\Omega(\log n) \leq f \leq n^{o(1)}$, we have almost matching lower and upper bounds of $n^{1 \pm o(1)}$ on the query time.
    
    In the regime $f = o(\log n)$, the query time in~\Cref{thm:main-cut-theorem} is $n^{o(1)}$ which trivially almost-optimal.
    Lastly, when $f= n^{\alpha}$ for some constant $\alpha \in (0,1]$, we ask if the query time of $O(fn)$ can be improved.
    As it turns out, the lower bound in~\cite{LongS22} shows that a polynomial improvement is unlikely, as it would refute the ``online version'' of the popular Orthogonal Vectors (OV) conjecture.
    This version has not been formally stated before, mainly because it did not find any concrete applications; nevertheless, it is considered plausible by the fine-grained complexity community~\cite{AbboudPersonal}.
    The details regarding these query time lower bounds appear in~\Cref{sec:condLB}.

\end{description}

\paragraph{On Incremental Updates.}
While our oracle focuses on decremental updates (namely, query connectivity upon the deletion of $f$ vertices), it also makes sense to consider the incremental setting, where one ``turns on'' $f$ vertices (the initial graph is given with some vertices turned on and some turned off).
Allowing both decremental and incremental updates is commonly known as the \emph{sensitivity} oracles setting, recently explored in the context of $st$ vertex cut oracles ~\cite{LongW24} where $\poly(f,\log n)$ query time can be obtained for $f$-size update.
Interestingly, for (global) vertex cut oracles, there is a strong separation between the decremental and incremental settings: using ideas from~\cite{HenzingerL0W17}, we prove that even for $f = O(1)$, the required query time is at least $n^{1-o(1)}$ in the incremental setting (assuming SETH).
The details appear in~\Cref{sec:incremental-sensitivity}.

\paragraph{Vertex Cut Labels.}
As for $f$-vertex cut labels, 
we essentially settle the questions posed by~\cite{ParterPP24,LongPS25} by providing a construction with label length that matches the known lower bound up to polylogarithimic factors:
\begin{theorem}\label{thm:vertex-cut-labels}
    For every $n$-vertex graph $G=(V,E)$ and integer $f = O(\log n)$, there is an $f$-vertex cut labeling with label length of $\tilde{O}(n^{1-1/f})$ bits. The total label length (summing over all vertices) is $\tilde{O}(n)$ bits.
    The labels are constructed deterministically in polynomial time. 
\end{theorem}
A nice aspect of the above labeling scheme is its (perhaps surprising) simplicity and black-box use of labels for the ``$st$ variant'' of the problem.
When $f=\Omega(\log n)$, a randomized scheme with label length $\tilde{O}(n)$ follows immediately by \cite{DoryP21}.
However, the best current deterministic solution for this regime has labels of $\tilde{O}(fn)$ bits, by using the standard graph sparsification of \cite{NagamochiI92} and storing the entire sparsified graph in the label of each vertex; we leave open the intriguing question of improving this trivial bound.

We observe that although the structure of minimum vertex cuts can be leveraged to create exponentially faster oracles, it does not offer any advantage in the context of labels. This follows by observing that the label length lower bound of \cite{LongPS25} holds also for $f$-connected graphs.

\paragraph{Structural Insights: ``Cut Respecting'' Expander Decomposition.} 
Many of the recent breakthrough algorithmic results for computing edge and vertex connectivity are based on some notion of hierarchical expander decomposition \cite{LiP20, Li21,HenzingerLRW24}. The first connection between expanders and the minimum edge cut problem has been observed by \cite{KawarabayashiT15,LiP20} and in the data-structure setting by \cite{PatrascuT07}. The high-level approach is based on the observation that the problem at hand (e.g., computing a minimum cut, constructing vertex cut oracles, etc.) is considerably simpler when the graph is an expander. The hierarchical expander decomposition provides a convenient machinery to reduce general graphs to expanders. See e.g., \cite{Li21,HenzingerLRW24}.
Edge- and vertex-expander hierarchies have also been used recently for the $st$ cut labeling schemes of~\cite{LongPS25}. 

Our $f$-vertex cut oracles are as well based on a variant of hierarchical expander decomposition admitting ``cut respecting'' properties.
Very informally, it shows that every graph can be decomposed into a collection of \emph{terminal expanders} of bounded total size, such that vertex cuts of size at most $f$ ``translate'' into terminal cuts in these expanders.
This structural decomposition is presented in~\Cref{sec:cut-respecting-family} (see~\Cref{def:f-cut-respecting-TED} for a formal description); we are hopeful it could have future applications in various contexts (e.g., the distributed and parallel settings).%

\paragraph{Subsequent work.} 
In a follow-up work,
\cite{Kosinas25} employed a DFS-based, rather than expander-based approach for vertex cut oracles\footnote{The oracles of \cite{Kosinas25} also report the number of connected components.} (it does not study vertex cut labels). For general graphs, their bounds shaves several logarithmic terms. For $f$-connected graphs (i.e., minimum vertex cut oracles), our query time and space bounds are considerably smaller, by factors of $f$ and $f^2$, respectively (and as mentioned, our space bound is near-optimal). 
Our work also provides a collection of (possibly conditional) lower bounds for the problem which establish the almost optimality of our result.

\paragraph{Organization.}
After a few preliminaries in~\Cref{sec:prelim}, the short, stand-alone~\Cref{sec:labels} provides our $f$-vertex cut labeling scheme of~\Cref{thm:vertex-cut-labels}.
We then move to consider $f$-vertex cut oracles, which comprises the majority of the paper.
The main
\Cref{sec:oracles} is devoted to proving~\Cref{thm:main-cut-theorem,thm:cut-theorem-fconnected}; a technical overview is given in the subsection~\ref{sec:overview-oracles} (a roadmap for the rest of the subsections appears at the end of the overview).
\Cref{sec:cut-respecting-family} provides the formal details of the structural ``cut respecting'' expander decomposition.
\Cref{sec:spaceLB} proves the (unconditional) space lower bound of~\Cref{thm:spaceLB}.
Conditional lower bounds are discussed in~\Cref{sec:condLB,sec:incremental-sensitivity}, pertaining $f$-vertex cut oracles and their ``incremental'' variant, respectively.

%% file: prelim.tex
\section{Preliminaries}\label{sec:prelim}
Let $G= (V,E)$ be an undirected $n$-vertex $m$-edge graph.
The neighbor-set of a vertex $v\in V$ is denoted by $N(v)$.
When $k \geq 0$ is some nonnegative integer, $N_k (v)$ denotes some fixed arbitrary subset of $N(v)$ of size $k$ (or $N_k (v) = N(v)$ if $|N(v)| \leq k$).
For $U \subseteq V$, we define $N (U) = (\bigcup_{u \in U} N(u)) - U$, i.e., $N(U)$ are those vertices outside $U$ with a neighbor in $U$.

A \emph{vertex cut} in $G$ is a partition $(L,S,R)$ of $V$ such that $L$ and $R$ are nonempty, and there are no edges going between $L$ and $R$.
The set $S$ is called the \emph{separator} of the cut.
Slightly abusing terminology, we call a vertex set $F \subseteq V$ a \emph{cut in $G$} if there exists some vertex cut in which $F$ is the separator.
Equivalently, $F$ is a cut in $G$ iff $G-F$ (the subgraph of $G$ induced on $V-F$) is a disconnected graph.
Let $T \subseteq V$ be some set of vertices in $G$ called \emph{terminals}.
A vertex cut $(L,S,R)$ is called a \emph{vertex $T$-cut} if $L \cap T \neq \emptyset$ and $R \cap T = \emptyset$.
In this case, we say that $S$ \emph{separates $T$ in $G$}. Equivalently, $S$ separates $T$ in $G$ iff there are two vertices $s,t \in T-S$ which are disconnected in $G-S$.
Again, slightly abusing terminology, we also call a vertex set $F \subseteq V$ a \emph{$T$-cut in $G$} if $F$ separates $T$.
We say that $G$ is a \emph{$(T,\phi)$-expander} with \emph{expansion} $0 < \phi \leq 1$, if for every vertex cut $(L,S,R)$ in $G$, it holds that $|S| \geq \phi \min \{|T \cap (L \cup S)|, |T \cap (R \cup S)|\}$.

The \emph{arboricity} of $G$ is the minimum number of forests into which its edges $E$ can be partitioned.
The \emph{$f$-connectivity certificates} of Nagamochi and Ibaraki~\cite{NagamochiI92} will be very useful for us.
\begin{theorem}[\cite{NagamochiI92}]\label{thm:conn-certificate}
    Let $f \geq 1$.
    Then $G$ has a subgraph $H = (V, E_H)$ of
    arboricity at most $f+1$ (in particular $|E_H| \leq (f+1)n$) with the following property:
    For every $F \subseteq V$ with $|F| \leq f$ and every $s,t \in V-F$, it holds that $s,t$ are connected in $H-F$ iff they are connected in $G-F$.
    The subgraph $H$ can be computed deterministically in $O(m)$ time.
\end{theorem}

%% file: cut-labels.tex
\section{Vertex Cut Labels}\label{sec:labels}

\paragraph{Building Blocks.}
A basic building block in our $f$-vertex cut labels of~\Cref{thm:vertex-cut-labels} are succinct \emph{$f$-vertex failure connectivity labels}; these are the ``$st$ cut'' variant of $f$-vertex cut labels.
\begin{theorem}[\cite{ParterPP24,LongPS25}]
    One can assign a label $\ell(v)$ of $\poly(f, \log n)$ bits to every $v \in V$ such that, for every $s, t \in V$ and $F \subseteq V$ with $|F| \leq f$, it is possible to determine if $s$ and $t$ are connected in $G-F$ by only inspecting the $\ell(\cdot)$ labels of $\{s,t\} \cup F$.
    The $\ell(\cdot)$ labels are constructed deterministically in polynomial time.
\end{theorem}
The above theorem is highly non-trivial, but luckily we only use it as a black-box.
Recall that in~\Cref{thm:vertex-cut-labels} we are interested in the regime $f = O(\log n)$, so the $\ell(\cdot)$ labels have $\tilde{O}(1)$ bits, which for us is essentially as succinct as unique vertex identifiers.
From now on, whenever we ``store a vertex $v$'', we mean writing its unique identifier and its $\ell(v)$ label, taking up only $\tilde{O}(1)$ bits.

Another building block is sparsification, a standard preliminary step in the area that is crucial for our approach here: by first applying the sparsification of~\cite{NagamochiI92} (formally stated in \Cref{thm:conn-certificate}), we may assume without loss of generality that $G$ has at most $(f+1)n$ edges.

\paragraph{Warm-up: Minimum Vertex Cuts (The $f$-Connected Case).}
We first focus on the special case where $G$ is $f$-connected, i.e., $G$ does not have any vertex cuts of size $f-1$ or less.
Our goal is to assign a label $L(v)$ of $\tilde{O}(n^{1-1/f})$ bits to every $v \in V$, so that for every $F = \{x_1, \dots, x_f\} \subseteq V$ we can determine if $F$ is a (minimum) vertex cut from the information stored in $L(x_1), \dots, L(x_f)$.
This case turns out to be extremely simple and easy to present while conveying most of the intuition leading to the general case.
The following claim is what makes the $f$-connected case so convenient:
\begin{claim}
    Let $F = \{x_1, \dots, x_f\}$ be a vertex cut in an $f$-connected graph $G$.
    Then every $x_i \in F$ has two neighbors that are separated by $F$.
\end{claim}
\begin{proof}
    Seeking contradiction, suppose all the neighbors of $x_i$ outside $F$ are in the same connected component $C$ of $G-F$ (if $N(x_i) \subseteq F$, choose $C$ arbitrarily),
    and let $D$ be a different connected component.
    Then $N(D) \subseteq F - \{x_i\}$ is a vertex cut in $G$ of size $< f$, a contradiction.
\end{proof}

This claim immediately yields labels where $L(v)$ with length $\tilde{O}(|N(v)|)$, by letting this label store all vertices in $\{v\} \cup N(v)$. 
To answer a query $F = \{x_1, \dots, x_f\}$, we can just check if there is a pair of vertices in any arbitrary $N(x_i)$ that are separated by $F$, using the $\ell(\cdot)$ labels of the vertices in $N(x_1)$ and of $x_1, \dots, x_f$.
If we don't find a separated pair, we can safely determine that $F$ is not a cut by the claim above.
However, $|N(v)|$ could be as large as $\Omega(n)$.

To overcome this, we partition the vertices into low- and high-degree, setting the threshold at $2(f+1)n^{1-1/f} = \tilde{\Theta}(n^{1-1/f})$.
Low-degree vertices still have the budget to store their entire neighborhoods. 
For the moment, we let high-degree vertices store just themselves.
If the query contains some low-degree vertex $x_i$, we can still employ the strategy above.
The remaining ``problematic'' queries are when $x_1, \dots, x_f$ all have high degrees.
Our threshold is set precisely to ensure that there are at most $n^{1/f}$ high-degree vertices in $G$.
Thus, each high-degree $v$ can take part in up to $(n^{1/f})^{f-1} = n^{1-1/f}$ problematic queries.
So, $L(v)$ has the budget to explicitly store a table with all these problematic queries and the required answers for them (``cut'' or ``not a cut'').
Now, given a problematic query $F= \{x_1, \dots, x_f\}$ of all high-degree vertices, we can just explicitly find the answer from the table of any arbitrary $x_i$.
This establishes~\Cref{thm:vertex-cut-labels} in the $f$-connected case.

\paragraph{General Graphs.}
We now turn to consider the general case, where $G$ is not necessarily $f$-connected.
Our construction still relies the low- and high-degree classification, with some additional important tweaks.
The high-level intuition is as follows.
By also letting the high-deg vertices store only $f+1$ of their neighbors, we can show the following dichotomy:
If $F$ is a cut, then either (i) at least two vertices among the stored neighbors of the query set $F$ are separated by $F$, or (ii) the subset of high-degree vertices in $F$ separates at least two vertices in $V-F$.
Case (i) is resolved using the $\ell(\cdot)$ labels, and case (ii) by exploiting that the number of high-degree vertices is bounded. We next describe the solution for general graphs in details.

\paragraph{Explicit Subsets.}
Vertex subsets $K \subseteq V$ of size $|K| \leq f$ consisting of high-degree vertices will be handled ``explicitly'' in our scheme, by inspecting $G-K$ and storing relevant information regarding it.
Let $A_K$ be the set of vertices in the connected component of $G-K$ with the maximal number of vertices (ties are broken arbitrarily).
Let $B_K$ be the union of the vertex sets of all other connected components in $G-K$.
We define the \emph{explicit label} $L(K)$ of $K$ as follows:
\begin{algorithm}[H]
\caption{Creating the explicit label $L(K)$ of a subset $K \subseteq V$ with $|K|\leq f$}\label{alg:explicit-label}
\begin{algorithmic}[1]
\State \textbf{store} $K$
\State \textbf{store} $|A_K|$
\If{$|A_K| \geq n-f$}
    \textbf{store} $B_K$
\EndIf
\end{algorithmic}
\end{algorithm}
Note that $L(K)$ consists of $O(f \log n)$ bits, because $|K| \leq f$, and $|B_K| \leq f$ when $|A_K| \geq n-f$.
The following lemma gives the reasoning behind the construction of $L(K)$:
\begin{lemma}\label{lem:explicit-cut}
    Let $K \subseteq F \subseteq V$ such that $|F| \leq f < n/2$.
    Suppose that either
    (i) $|A_K| < n - |F|$, or
    (ii) $|A_K| \geq n -|F|$ and $B_K \not \subseteq F$.
    Then $F$ is a vertex cut in $G$.
\end{lemma}
\begin{proof}
    If $G-F$ is connected, then all pairs of vertices in $V-F$ are connected via paths that avoid $F \supseteq K$, hence $|A_K| \geq |V - F| = n - |F|$.
    Thus, if (i) holds, $F$ must be a vertex cut in $G$.
    If (ii) holds, then $B_K$ contains some $v \notin F$.
    Also, $A_K$ contains some $u \notin F$, as otherwise we would have $|F| \geq |A_K| \geq n-|F|$, implying that $f \geq |F| \geq n/2$, contradicting $f < n/2$.
    The vertices $v$ and $u$ lie in different connected components of $G-K$, and hence also of $G-F$, so $F$ is a cut in $G$.
\end{proof}

\paragraph{Constructing the Labels.}
We say that a vertex $v$ has \emph{low degree} if $|N(v)| \leq 2(f+1)n^{1-1/f}$ and \emph{high degree} otherwise.
Let $D$ and $H$ denote the sets of low-degree and high-degree vertices, respectively.
As $G$ has at most $(f+1)n$ edges, we get $|H| \leq n^{1/f}$.
Our labels are defined differently for low-degree and for high-degree vertices:

\begin{algorithm}[H]
\caption{Creating the label $L(v)$ of a $v \in V$}\label{alg:vertex-label}
\begin{algorithmic}[1]
\State \textbf{store} $v$ and $\ell(v)$
\If{$v \in D$}
    \State \textbf{store} $u$ and $\ell(u)$ for every $u \in N(v)$
\EndIf
\If{$v \in H$}
    \State \textbf{store} $u$ and $\ell(u)$ of every $u \in N_f(v)$
    \State \textbf{store} $L(K)$ of every $K \subseteq H$ such that $v \in K$ and $|K| \leq f$.
\EndIf
\end{algorithmic}
\end{algorithm}
To analyze the label length,
recall that each $\ell(u)$ label and each $L(K)$ label has only $\poly(f, \log n) = \tilde{O}(1)$ bits.
If $v \in D$, then $L(v)$ has $\tilde{O}(|N(v)|) \leq \tilde{O}(n^{1-1/f})$ bits.
If $v \in H$, then the number of subsets $K \subseteq H$ with $v \in K$ and $|K| \leq f$ is $\binom{|H|}{\leq f-1} \leq (f-1) \cdot |H|^{f-1} = O(f n^{1-1/f})$.
Also, $|N_f (v)| \leq f$.
So again, $L(v)$ has $\tilde{O}(n^{1-1/f})$ bits.
We next bound the total label length by $\tilde{O}(n)$.
Indeed,
low-degree vertices contribute up to $\sum_{v \in D}|L(v)|=\tilde{O}(|N(v)|) \leq O(|E|) = \tilde{O}(n)$ bits (as $G$ has $\leq (f+1)n = \tilde{O}(n)$ edges),
and high-degree vertices contribute up to $|H| \cdot \tilde{O} (n^{1-1/f}) = \tilde{O}(n)$ bits.

Lastly, we address construction time.
We only need to construct $L(K)$ for sets $K \subseteq H$ with $|K| \leq f$, which are at most $\tilde{O}(f |H|^f) = \tilde{O}(n)$.
Computing a single $L(K)$ clearly takes polynomial time.
Given the needed $L(K)$'s, computing all $L(v)$ for $v \in V$ also takes polynomial time.

\paragraph{Answering Queries.}
We now present the query algorithm, that given the $L(\cdot)$-labels of the vertices in $F \subseteq V$, $|F| \leq f$, determines if $F$ is a vertex cut in $G$.
It is given as~\Cref{alg:query}.
\begin{algorithm}[H]
\caption{Answering a query $F \subseteq V$, $|F| \leq f$ from the labels $\{L(v) \mid v \in F\}$}\label{alg:query}
\begin{algorithmic}[1]
\State let $T := \{u \in V-F \mid \text{$\ell(u)$ is stored in some $L(v)$ of $v \in F$}\}$
\State choose an arbitrary $s \in T$
\For{every $t \in T$}
    \State use the $\ell(\cdot)$ labels of $\{s,t\}\cup F$ to check if $s,t$ are connected in $G-F$.
    \If{$s,t$ are disconnected in $G-F$}
        \Return ``cut''
    \EndIf
\EndFor
\If{$K := H \cap F \neq \emptyset$}
    \State find $L(K)$ in any $L(v)$ of $v \in K$
    \State find $|A_K|$ in $L(K)$
    \If{$|A_K| < n-|F|$}{
        \Return \emph{``cut''}
    }
    \Else{
        \State find $B_K$ in $L(K)$
        \If{$B_K \not \subseteq F$}{
            \Return \emph{``cut''}
        }
        \EndIf
    }
    \EndIf
\EndIf
\State \Return \emph{``not a cut''}
\end{algorithmic}
\end{algorithm}

\paragraph{Correctness.}
The soundness direction is straightforward: if \emph{``cut''} is returned, then either we found $s,t \in V-F$ which are disconnected in $G-F$, or otherwise we have found some $K \subseteq F$ which certifies that $F$ is a vertex cut in $G$ by \Cref{lem:explicit-cut}.
The completeness direction remains: assuming $F$ is a cut, we should prove that \emph{``cut''} is indeed returned.
Let $(A,F,B)$ be a vertex cut in $G$.
Thus, $N(A) \subseteq F$ and $N(B) \subseteq F$.
We consider two complementary cases:

\begin{description}
\item[Case 1:]
There is some $x \in N(A)-N(B)$ and some $y \in N(B)-N(A)$.

We will show that then, the set $T$ in~\Cref{alg:query} must contain a vertex from $A$ and a vertex from $B$, so one of them must be disconnected from $s$ in $G-F$, hence \emph{``cut''} is returned.
We prove that $T \cap A \neq \emptyset$; the proof that $T \cap B \neq \emptyset$ is symmetric, by replacing $x$ with $y$ and swapping $A$ and $B$ everywhere in the following argument.

By the definition of $T$ and of the label $L(x)$, we have $N_f (x) - F \subseteq T$.
We claim that $x$ must have some neighbor $t \in N_f (x) - F$;
otherwise, we get $N_f (x) \subseteq F-\{x\}$, so $|N_f (x)| \leq f-1$, hence $N_f (x) = N(x)$, thus $N(x) \subseteq F$, contradicting that $N(x) \cap A \neq \emptyset$ as $x \in N(A)$.
We now observe that $t \in A$, as $t \notin F$, and $t \notin B$ because $x \notin N(B)$.
So, $t \in T \cap A$ as required.

\item[Case 2:]
One of $N(A)$, $N(B)$ is contained in the other.
Without loss of generality, $N(B) \subseteq N(A)$.
\begin{description}
    \item[Case 2a:]
    $N(B)$ contains a low-degree vertex $v \in D$.
    
    Since $v \in N(B) \subseteq N(A)$, $v$ must have a neighbor from $A$ and a neighbor from $B$.
    Since $v \in D$, the definition of $L(v)$ and of $T$ implies that $N(v) \subseteq T$.
    Thus, $T$ contains a vertex from $A$ and a vertex from $B$, hence \emph{``cut''} is returned (as argued in Case 1).

    \item[Case 2b:]
    $N(B)$ contains only high-degree vertices. Thus, $N(B) \subseteq K = H \cap F$.
    
    If $|A_K| < n-|F|$, we return \emph{``cut''} as needed.
    If $|A_K| \geq n-|F|$, we show that $B_K \not \subseteq F$, hence we also return \emph{``cut''}.
    Seeking contradiction, suppose $B_K \subseteq F$.
    Then $A \cup B = V-F \subseteq V-(K \cup B_K) = A_K$.
    Choose $s \in A$, $t \in B$, and a path $P$ from $s$ to $t$ in $G[A_K]$ which exists as $G[A_K]$ is a connected component of $G-K$.
    Then $P$ starts in $V-B$ and ends in $B$, so $V(P) \cap N(B) \neq \emptyset$.
    But, $V(P) \cap N(B) \subseteq A_K \cap K = \emptyset$ --- contradiction.
\end{description}

\end{description}
This concludes the proof of~\Cref{thm:vertex-cut-labels}.

%% file: overview-oracles.tex
\subsection{Technical Overview}\label{sec:overview-oracles}

\paragraph{Terminal Reduction.}
Our construction adapts the \emph{terminal reduction} framework of~\cite{NSYArxiv23}, originally devised for a sequential algorithm that outputs one minimum vertex cut, to the data structure setting of $f$-vertex cut oracles that need to handle any query of a potential vertex cut (which might not be a minimum cut).
To facilitate this adaptation, we introduce the notion of \emph{terminal cut detectors}, which are relaxed variants of $f$-vertex cut oracles, on which the main terminal reduction result is based.

\begin{definition}[Terminal Cut Detector]\label{def:cut-detector}
    Let $G=(V,E)$ be a graph with two terminal sets $T,S \subseteq V$, and let $f \geq 1$.
    An \emph{$(f,T,S)$-cut detector} for $G$ is a data structure $\D$ that can be queried with any $F \subseteq V$ s.t.\ $|F| \leq f$, and returns either \emph{``cut''} or \emph{``fail''} with the following guarantees:
    \begin{itemize}
        \item
        \emph{Soundness:}
        If $\D(F)$ returns \emph{``cut''}, then $F$ is a cut in $G$.
        
        \item
        \emph{Completeness:} 
        If $F$ separates $T$ but does not separate $S$, then $\D(F)$ returns \emph{``cut''}. In other words, if $\D(F)$ returns \emph{``fail''} and $F$ separates $T$, then $F$ also separates $S$.
    \end{itemize}
\end{definition}

Our key construction is a procedure that is given a terminal set $T$ and outputs a smaller terminal set $S^*$ as well as an $(f,T,S^*)$-cut detector. 
\begin{theorem}[Terminal Reduction]\label{thm:cut-detectors}
    There is a deterministic algorithm that given an $n$-vertex graph $G=(V,E)$ with terminal set $T \subseteq V$ and integer parameter $f = O(\log n)$, computes
    \begin{itemize}
        \item a new terminal set $S^* \subseteq V$ such that $|S^*| \leq \frac{1}{2}|T|$, and
        \item an $(f,T,S^*)$-cut detector $\D$ with space $\tilde{O}(n)$ and query time $\tilde{O}(2^{|F|})$.
    \end{itemize}
    Assuming $G$ has $O(fn)$ edges, the running time can be made $\tilde{O}(n^{1+\delta})$ for any constant $\delta > 0$.
    It can be improved to $n^{1+o(1)}$ at the cost of increasing the space and query time of $\D$ by $n^{o(1)}$ factors.
\end{theorem}

Given~\Cref{thm:cut-detectors}, we can easily derive the $f$-vertex cut oracles of~\Cref{thm:main-cut-theorem}:
\begin{proof}[Proof of~\Cref{thm:main-cut-theorem}]
    Define $T_0 = V$. For $i = 1, 2, \dots$, apply~\Cref{thm:cut-detectors} with input terminal set $T_{i-1}$, denoting the new  terminal set by $T_i$ and the resulting  $(f, T_{i-1}, T_i)$-cut detector by $\D_{i-1}$.
    Halt after the first iteration $\ell$ that produced $T_\ell = \emptyset$; we have $\ell = O(\log n)$ as the size of the new terminal set halves in each iteration.
    The $f$-vertex cut oracle simply stores $\D_0, \D_1, \dots, \D_{\ell-1}$.
    
    Given query $F \subseteq V$ s.t.\ $|F| \leq f$, the oracle queries each $\D_i$ with $F$, and determines that $F$ is a vertex cut in $G$ if and only if some $\D_i (F)$ returned \emph{``cut''}.
    Indeed, if $F$ is not a cut in $G$, the soundness in~\Cref{def:cut-detector} ensures that no $\D_i (F)$ returns \emph{``cut''}. 
    Conversely, if $F$ is cut in $G$, then let $i$ be the largest index such that $F$ separates $T_i$ (note that $0 \leq i \leq \ell-1$ as $F$ separates $T_0 = V$ but not $T_\ell = \emptyset$); 
    by the completeness in~\Cref{def:cut-detector}, $\D_i (F)$ returns \emph{``cut''}.

    The parameters (space, query and preprocessing time) of the resulting $f$-vertex cut oracle are only larger by an $\ell = O(\log n)$ factor then those in~\Cref{thm:cut-detectors}, which yields~\Cref{thm:main-cut-theorem}.%
    \footnote{
        The additive $O(m)$ term in preprocessing time of~\Cref{thm:main-cut-theorem} is for applying the standard sparsification of~\cite{NagamochiI92} to reduce the number of edges in $G$ to $O(fn)$.
    }
\end{proof}

So, from now on we set our goal to prove~\Cref {thm:cut-detectors}. We start by considering the case where the given graph is a terminal expander for which a solution follows readily by using the data-structure for $st$-cut variant. Then, we explain how to handle general graphs by reducing to terminal expanders.  

\paragraph{Warm-up: Terminal Expanders (or Few Terminals).}
We start by observing that~\Cref{thm:cut-detectors} is easy to prove when $G$ is a $(T,\phi)$-expander.
In fact, in this case we can just take $S^* = \emptyset$, and reduce the query time to $\poly(f, \log n) = \tilde{O}(1)$.
As usual, in the following the graph $G=(V,E)$ has $n$ vertices and $m$ edges; The notation $\bar{m}$ stands for $\min\{m, fn\}$.

\begin{restatable}[Terminal Expanders]{lemma}{terminalexpanders}\label{thm:terminal-expander-DS}
    Suppose $G$ is known to be a $(T,\phi)$-expander for $T \subseteq V$ and $0 < \phi \leq 1$.
    Then, there is a deterministic $(f,T,\emptyset)$-cut detector for $G$, denoted $\TEOracle(G,T,\phi,f)$, with
    $\tilde{O}(\bar{m})$ space, $\tilde{O}(f^2/\phi)$ query time, and $m^{1+o(1)}/\phi + \tilde{O}(f\bar{m})$ preprocessing time.
\end{restatable}
The proof is based on the fact that the terminal expander $G$ admits a \emph{low-degree Steiner tree} for $T$; this is a tree $\tau$ in $G$ that connects all terminals in $T$ and has maximum degree $\tilde{O}(1/\phi)$, as was shown in~\cite{LongS22full,LongPS25}.
Consider a query $F$, and let $N_\tau (F)$ be the set of neighboring vertices to $F$ on the tree $\tau$.
Then it suffices for us to decide whether all vertices in $N_\tau (F)$ are connected in $G-F$: if so, then $F$ does not separate $T$, and if not, then $F$ is a cut in $G$.
So answering the query $F$ amounts to checking connectivity between $\tilde{O}(|F|/\phi)$ vertices in $G-F$.
This can be done efficiently by using an \emph{$f$-vertex failure connectivity oracle}~\cite{DuanP20,LongS22,LongW24} (i.e., $st$-cut oracles) which can be updated with a failure set $F \subseteq V$ s.t.\ $|F| \leq f$, and subsequently can answer connectivity queries between pairs of vertices in $G-F$.

In fact, $f$-vertex failure connectivity oracles also immediately yield trivial $(f,T,\emptyset)$-cut detectors, by updating with the query $F$ and checking connectivity in $G-F$ between all vertices in $T$ (i.e., choosing an arbitrary terminal and checking if all other terminals are connected to it).
This strategy is efficient when there are only few terminals.
\begin{restatable}[Few Terminals]{lemma}{fewterminals}\label{lem:terminal-data-structure}
    Let $T \subseteq V$.
    There is a deterministic $(f,T,\emptyset)$-cut detector for $G$, denoted $\FEOracle(G,T)$, with
    $\tilde{O}(\bar{m})$ space, $\tilde{O}(f^2) + O(f|T|)$ query time, and $O(m) + \bar{m}^{1+o(1)} + \tilde{O}(f\bar{m})$ preprocessing time.
\end{restatable}

The formal proofs of~\Cref{thm:terminal-expander-DS} and~\Cref{lem:terminal-data-structure} are in~\Cref{sec:special-cut-detectors}.
Of course, in the general case of~\Cref{thm:cut-detectors}, $G$ might not be a terminal expander and the terminal set $T$ could be very large; the high-level strategy is recursively decomposing this general instance into many instances of terminal expanders (or graphs with few terminals), as we discuss next.

\paragraph{The Left-Right Decomposition.}
On a high level, the divide-and-conquer decomposition approach is based on splitting the original instance $(G,T)$ into two, according to a \emph{sparse} vertex $T$-cut $(L,S,R)$, namely, such that $|S|$ is much smaller than the number of terminals in $L \cup S$ or in $R \cup S$; such a cut exists as $G$ is not a terminal expander.
The splitting technique was introduced in~\cite{SaranurakY22} for the case $T=V$, and extended to $T \subseteq V$ in~\cite{NSYArxiv23} (we further carefully adapt it to handle \emph{non-minimum} cuts).

The two new instances are called the \emph{$f$-left and $f$-right graphs} $G_L$ and $G_R$.
The $f$-left graph $G_L$ is (except in corner cases) just $G[L \cup S]$ plus an additional ``representative set'' $U_R$ of $f+1$ terminals from $R$, that are connected by a clique between themselves and by a biclique to $S$. $G_R$ is defined symmetrically.
Note that $G_L$ and $G_R$ are not necessarily subgraphs of $G$, but their vertex sets is a subset of $V(G)$. These graphs have two crucial properties that we term as \emph{cut-respecting properties}: 
The ``completeness'' property is that together, $G_L$ and $G_R$ capture every $T$-cut $F$ s.t.\ $|F| \leq f$, \emph{as long as $F$ does not separate $S$}.
Namely, for such $F$, either the remaining part of $F$ in $G_L$ separates the remaining terminals there, or this will happen in $G_R$.
Additionally, they also have ``soundness'' properties, ensuring that a vertex cut in $G_L$ or $G_R$ of size $\leq f$ is also a cut in $G$ (since $V(G_L), V(G_R)\subseteq V(G)$, this is well-defined).

Roughly speaking and ignoring some technical nuances, we keep on decomposing each of $G_L$ and $G_R$ recursively, until we get terminal expanders or instances with few terminals.
The required new terminal set $S^*$ for~\Cref{thm:cut-detectors} is then defined as the union of separators $S$ over all the sparse cuts that were found across the recursion.
The cut-respecting properties imply that if $F$ is a separator of $T$ of size $\leq f$ that is not captured by any leaf instance of the recursion, then $F$ must also separate $S^*$, and so an $(f,T,S^*)$-cut detector can afford to return \emph{``fail''} on $F$.

The recursion is made effective and efficient by using a powerful tool devised in~\cite{LongS22} that finds a ``good'' sparse cut $(L,S,R)$.
Intuitively, this sparse cut is promised to either cause the number of terminals to shrink (by some constant factor) in both $G_L$ and $G_R$, or to cause terminal shrinking on one side while ensuring the other is an expander that requires no recursion.

We call the resulting binary recursion tree $\T$ the $f$-left-right decomposition tree ($f$-LR tree, for short); storing the recursive instances in this tree serves the basis of our $(f,T,S^*)$-cut detector $\D$.
The $f$-LR tree $\T$ is shown to have the following key properties.
\begin{itemize}
    \item (Logarithmic depth) The depth of $\T$ is $O(\log |T|)$.
    \item (Near-linear space) Storing all the instances $(G',T')$ in the nodes of $\T$ consumes $\tilde{O}(n)$ space.
    \item (Terminal reduction) The union $S^*$ of all separators $S$ from the sparse cuts $(L,S,R)$ found in the internal nodes of $\T$ is of size at most $|T|/2$.
\end{itemize}

Recall that the leaf instances in $\T$ are terminal expanders or have few terminals, for which we can apply the efficient cut detectors of~\Cref{thm:terminal-expander-DS} and~\Cref{lem:terminal-data-structure}, respectively. For a given $F \subseteq V$ s.t.\ $|F| \leq f$, a natural query algorithm simply feeds $F$ into all cut detectors in the leaves. The correctness is immediate by the cut-respecting properties of the left-right decomposition. Specifically, if $F$ separates $T$ but not $S^*$, then one of these must return cut \emph{``cut''};
and conversely, if one of them returns\emph{``cut''}, then $F$ must be a cut in $G$. The key limitation of this approach is in the query time which might be linear as there might be $\Omega(n)$ leaves in $\T$ that are explored. To obtain the desired $\tilde{O}(2^{|F|})$ query time, we devise a tree-searching procedure that essentially allows us to focus on only $O(2^{|F|})$ leaves.

\paragraph{Tree Searching.}
Let us fix some query $F \subseteq V$ s.t.\ $|F| \leq f$.
Consider first the root node of $\T$, which is associated with the original instance $(G,T)$.
The cut-respecting properties guarantee us that if $F$ separates $T$ but not $S^*$, then it is captured either by $G_L$ in the left child or by $G_R$ in the right child.
But we cannot tell a priori which is the correct child, so it seems like we have to ``branch'' our search recursively to both the left and right subtrees of $\T$.
As this phenomenon might reoccur in many nodes in our search, this could lead us to visit too many nodes in $\T$.
However, in some cases, we do not have to branch.
For example, what if $F$ does not contain any vertex of $G_L$?
Then clearly, there is no point exploring the left child; we can ``trim'' the search on the left subtree and continue our search only in the right subtree.
So intuitively, we want to trim as many recursive searches as we can to get faster query time.

Let $\T(F)$ denote the subtree of $\T$ which is visited during the tree search for the query $F$.
The ``branch'' nodes are those in $\T(F)$ that have two children, and the ``trim'' nodes are those that have only one child.
Our search mechanism ensures enough trimming so that the following property holds:
If $(G',T')$ is an instance in some ``branch'' node, then $|F \cap V(G'_L)|,|F\cap V(G'_R)| \leq |F \cap V(G')|-1$.
In other words, when moving from a ``branch'' node to its children, the size of the query decreases by at least $1$ in both the left and the right child.
As the original query at the root has size $|F|$, this implies that there are only $O(2^{|F|})$ ``branch'' nodes.
So, as $\T(F)$ only has depth $O(\log |T|)$, it can only contain $O(2^{|F|} \log |T|)$ nodes.

Consider again the root of $\T$, and suppose now that $F \subseteq V(G_L)$.
We have to trim one of the subtree searches, as otherwise, the root would be a ``branch'' node violating the required property.
Because $F$ is consumed entirely in the left graph $G_L$, trimming the search in the right graph $G_R$ seems more plausible.
Note that $G_L$ and $G_R$ share only the vertices in $S$ and in the representative terminal sets of the sides $U_L$ and $U_R$.
Thus, the part $F \cap V(G_R)$ of $F$ that survives in $G_R$ must be contained in $U_L \cup U_R \cup S$.
As we only care about detecting vertex cuts that do not separate $S^* \supseteq S$, this special structure of $F \cap V(G_R)$ turns out to be enough for us to answer this query in $G_R$ directly, without any further recursion.
We term the relevant data structure to handle this case the ``US cut detector'' (where US stands for $U = U_L \cup U_R$ and $S$):

\begin{restatable}[US Cut Detectors]{lemma}{uscutdetector}\label{lem:US-data-structure}
    Let $G = (V,E)$ be a graph with $n = |V|$ and $m = |E|$.
    Let $U, S \subseteq V$.
    There is an $(f,V,S)$-cut detector $\USOracle(G,U,S,f)$ which is restricted to answer queries $F \subseteq S \cup U$, $|F| \leq f$, with
    $\tilde{O}(2^{|U|} fn)$ space, $O(2^{|F|} f \log n)$ query time, and $O(2^{|U|} (m + fn\log n))$ preprocessing time.
    Further, if $G$ is known to be $f$-connected, the query time improves to $O(f \log n)$.
\end{restatable}

The proof is found in~\Cref{sec:special-cut-detectors}.
So, by augmenting the root with $\USOracle(G_R, U_L \cup U_R, S, f)$, we can just query this cut detector with $F \cap V(G_R)$ and trim the recursive search in the right subtree.
Of course, there is nothing special about the root, or about its right child; we augment all nodes in $\T$ with US cut detectors for trimming searches in their left or right children.

\paragraph{Improving Space and Preprocessing Time.}
There is one small caveat with the approach described above: the $O(2^{|U|})$ factors in the space and preprocessing time of~\Cref{lem:US-data-structure} lead to corresponding factors of $O(2^{2f})$ in the space and preprocessing time of $\D$ in~\Cref{thm:cut-detectors}, and hence of the entire oracle in~\Cref{thm:main-cut-theorem} (as we use U-sets of the form $U_L \cup U_R$, that have size up to $2f+2$).
On a high level, these are shaved by applying ``hit-miss hashing'' tools~\cite{KarthikP21}, allowing us to essentially focus only on the case where query $F$ does not contain any terminals from $T$.
As $U_L \cup U_R \subseteq T$ (except in technical corner cases), this means that whenever we use a US cut detector, the query is actually contained in the S-set. Thus, we can construct this oracle with an \emph{empty} U-set instead.

\paragraph{Minimum Vertex Cuts.}
The oracle of \Cref{thm:cut-theorem-fconnected} is obtained by slightly tweaking the general construction to take advantage of the structural properties of minimum vertex cuts.

When the graph $G$ is $f$-connected, one can show an even stronger ``completeness'' property of the left and right graphs: If $F$ separates $T$ but not $S$, then either (1) $F \subseteq L \cup S$ and $F$ separates $T$ in $G_L$, or (2) $F \subseteq R \cup S$ and $F$ separates $T$ in $G_R$.
This property allows us to completely eliminate the need to branch in the tree-searching process for query $F$.
First, if $F$ intersects both $L$ and $R$, then it cannot be that $F$ separates $T$ but does not separate $S^* \supseteq S$, so we can safely return \emph{``fail''}.
Next, if $F$ intersects $L$ and not $R$ (resp., $R$ and not $L$), we only have to recursively search the left (resp., right) subtree.
Otherwise, we have $F \subseteq S$, and then we can use US cut detectors to trim the searches in both the left and right subtrees.
Thus, the resulting tree search process visits only a path in $\T$, which eliminates the $2^f$ factor in the query time that appeared in the general case.

We remark that in fact, it suffices for the query to be a \emph{``minimal''} cut in $G$ (in a properly defined sense) for the above arguments to work, even if the graph $G$ is not $f$-connected.

\paragraph{Roadmap.}
We first provide the cut detectors for special cases (terminal expanders, few terminals, US) in~\Cref{sec:special-cut-detectors}.
Next, we formally discuss left and right graphs and their properties in~\Cref{sec:left-right-graphs}.
The following~\Cref{sec:left-right-tree} presents the recursive left-right decomposition producing the tree $\T$ and the new terminal set $S^*$.
Then, \Cref{sec:construction-and-query-of-D} formally discusses the construction of and query algorithm of the $(f,T,S^*)$-cut detector $\D$ of~\Cref{thm:cut-detectors} (without shaving the $O(2^{2f})$ factors in preprocessing time and space).
The modifications that yield improved results when $G$ is $f$-connected are described next, in~\Cref{sec:f-connected}.
Finally, \Cref{sec:space-improvment} explains how the $O(2^{2f})$ factors in the preprocessing and query time from~\Cref{sec:construction-and-query-of-D} can be shaved.

%% file: special-cut-detectors.tex
\subsection{Special Cut Detectors}\label{sec:special-cut-detectors}

This section presents some special cut detectors, which essentially form the ``base cases'' of the recursive approach that used to create the cut detectors for general graphs of~\Cref{thm:cut-detectors}, as overviewed in~\Cref{sec:overview-oracles}.
Throughout this section, $G = (V,E)$ is a connected graph with $n$ vertices and $m$ edges, $f \geq 1$ is an integer parameter, and $\bar{m} = \min\{m, fn\}$.

As a basic building block, we use efficient \emph{$f$-vertex failure connectivity oracles} as a black box.
Such an oracle for $G$ can be \emph{updated} with any set $F \subseteq V$ of size $|F| \leq f$.
Once updated, the oracle can receive \emph{queries} of vertex pairs $(s,t) \in V \times V$, which are answers by indicating whether $s,t$ are connected in $G-F$.
We will use the current state-of-the-art, given by \cite{{LongW24}}:

\begin{theorem}[\protect{\cite[Theorem 1]{LongW24}}]\label{thm:storacle}
    There is an $f$-vertex failure connectivity oracle for $G$ with $\tilde{O}(\bar{m})$ space, $\tilde{O}(f^2)$ update time, and $O(f)$ query time.
    The oracle can be constructed deterministically in $O(m) + \bar{m}^{1+o(1)} + \tilde{O}(f\bar{m})$ time.
\end{theorem}

As a direct corollary, we get~\Cref{lem:terminal-data-structure} concerning the ``few terminals'' case.
\fewterminals*
\begin{proof}
    We simply use the $f$-vertex failure oracle connectivity oracle $\mathcal{O}$ of~\Cref{thm:storacle}, and also store $T$.
    Given a query $F \subseteq V$ with $|F| \leq f$, we first update $\mathcal{O}$ with $F$, and then apply $|T-F|-1$ connectivity queries $(s,t)$, where $s$ is some fixed terminal in $T-F$, and $t$ ranges over the remaining terminals from $T-F$.
    Note that $F$ separates $T$ iff $s$ is disconnected from some other $t \in T-F$, in which case we answer \emph{``cut''}; otherwise we answer \emph{``fail''}.
\end{proof}

Next, we consider the case of terminal expanders in~\Cref{thm:terminal-expander-DS}:
\terminalexpanders*
\begin{proof}
    As $G$ is a $(T,\phi)$-expander, we can apply the deterministic, $(m^{1+o(1)}/\phi)$-time algorithm of \cite[Lemma 5.7]{LongS22full}, that computes a Steiner tree $\tau$ for $T$ with maximum degree $O(\log^2 n / \phi)$.%
    \footnote{The degree bound can be improved to $O(1/\phi)$ at the cost of spending $O(mn \log n)$ time to compute $\tau$, as shown in~\cite{LPSSODA25-arXiv}.
    Since we incur some logarithmic factors anyway, we prefer the faster running time.}
    Additionally, we compute the oracle $\mathcal{O}$ of~\Cref{thm:storacle} for $G$.
    The cut detector $\TEOracle(G,T,f)$ simply stores $\tau$ and $\mathcal{O}$.
    The stated space and prepossessing time follow immediately.

    Given a query $F \subseteq V$ with $|F| \leq f$, we first update $\mathcal{O}$ with the set $F$ as the faults.
    Next, we compute the $\tau$-neighbors of $F$, namely all the vertices appearing in $N_\tau (x)$ for some $x \in F$ (where $N_\tau (x) = \emptyset$ if $x \notin V(\tau)$).
    We choose one arbitrary non-faulty vertex $v^* \notin F$ which is $\tau$-neighbor of $F$ (if no such exist, we return \emph{``fail''}).
    Finally, we query $\mathcal{O}$ with $(v^*,u)$ for every other non-faulty $u \notin F$ which is a $\tau$-neighbor of $F$.
    If one of these queries returned that $v^*,u$ are disconnected in $G-F$, we return \emph{``cut''}; otherwise we return \emph{``fail''}.
    
    The soundness of this procedure is clear, as we only return \emph{``cut''} in case we find a pair of disconnected vertices in $G-F$.
    For the completeness, suppose $F$ separates $T$, and let $t,t' \in T - F$ be two terminals separated by $F$.
    Consider the path $P$ in $\tau$ between $t$ and $t'$.
    As $F$ separates $t$ from $t'$, $F$ must intersect $P$.
    Let $x$ and $x'$ be the vertices in $F \cap V(P)$ that are closest to $t$ and $t'$, respectively.
    Let $u$ and $u'$ be the neighbors of $x$ and $x'$ in the directions of $t$ and $t'$, respectively.
    Note that $u$ and $u'$ are different (even if $x=x'$), and are both non-faulty $\tau$-neighbors of $F$.
    Also, the segments of $P$ between $t$ and $u$, and between $t'$ and $u'$, are both present in $G-F$.
    Hence, as $t$ and $t'$ are disconnected in $G-F$, so are $u$ and $u'$.
    Therefore, one of the queries $(v^*,u)$ and $(v^*,u')$ to $\mathcal{O}$ must return that the pair is disconnected in $G-F$, so we return \emph{``cut''} on $F$, as required.

    Finally, we address the query time.
    Updating $\mathcal{O}$ with $F$ takes $\tilde{O}(f^2)$ time, and then each subsequent query to $\mathcal{O}$ takes $O(f)$ time.
    As the maximum degree in $\tau$ is $O(\log^2 n/\phi)$, $F$ can have at most $O(f \log^2 n/\phi)$ of $\tau$-neighbors, which means that there are $O(f \log^2 n/\phi)$ queries to $\mathcal{O}$.
    So we get a total query time of $\tilde{O}(f^2/\phi)$. 
\end{proof}

We end this section by describing how US cut detectors are constructed.
This relies on the following short structural lemma.
\begin{lemma}\label{lem:US-lemma}
    Let $S \subseteq V$.
    Then for every $F \subseteq V$,
    $F$ is a cut in $G$ if and only if (at least) one of the following options holds:
    \begin{enumerate}
        \item $F$ separates $S$ in $G$.
        \item $F \supseteq S$ and $G-(S \cup F)$ is disconnected.
        \item $F \not \supseteq S$ and $G-(S \cup F)$ has a connected component $C$ with $N(C) \subseteq F$ (here $N(C)$ refers to the neighbor set of $C$ in $G$).
    \end{enumerate}
\end{lemma}
\begin{proof}
    The ``if'' direction: 
    If option 1 holds then of course $F$ is a cut.
    If option 2 holds, then $G-F = G-(S\cup F)$ and the latter is disconnected, so $F$ is a cut.
    If option 3 holds, then $F$ separates some $u \in S-F$ from any $v \in C$, so $F$ is a cut. 

    The ``only if'' direction:
    Assume $F$ is a cut.
    If option 1 holds we are done, so suppose $F$ does not separate $S$.
    If $F \supseteq S$, then $G-(S\cup F) = G-F$ and the latter is disconnected, so option 2 holds.
    If $F \not \supseteq S$, let $D$ be the connected component of $G-F$ such that $S-F \subseteq D$ (which exists since $F$ does not separate $S$), and 
    let $C \neq D$ be a different connected component of $G-F$ (which exists since $F$ is a cut).
    Then clearly $N(C) \subseteq F$.
    Note that $C$ does not intersect $S$, since $S \subseteq F \cup D$.
    Thus, $C$  remains connected in $(G-F)-S = G-(S \cup F)$, so it must also be a connected component of $G-(S \cup F)$.
    Namely, we have shown that option 3 holds.
\end{proof}

We are now ready to (restate and) prove~\Cref{lem:US-data-structure}:
\uscutdetector*
\begin{proof}%
    First, $\USOracle(G,U,S,f)$ stores $U$ and $S$.
    Next, for every $W \subseteq U$, it stores the collection of all neighbor-sets $N(C)$ of connected components $C$ in $G-(S \cup W)$ such that $|N(C)| \leq f$.
    (Again, here $N(C)$ refers to the neighbor set of $C$ in $G$.)
    This collection is stored in a \emph{sorted} array, where each set is represented by a bitstring of length $O(f \log n)$, formed by concatenating the IDs of the vertices in the set (in sorted order).
    Additionally, it stores a single bit which indicates if $G-(S \cup W)$ is connected.
    The space needed for a specific $W \subseteq U$ is dominated by its corresponding array, which takes up $O(fn)$ words.
    The preprocessing time needed for $W$ is $O(m + fn \log n)$, as the collection of relevant sets $N(C)$ is easy to compute in $O(m)$ time, and sorting the array takes $O(f \cdot n \log n)$ time (again, we multiply by $f$ to account for the comparison between elements of $f$ words).
    Summing over $W$ gives the desired space and preprocessing bounds.

    To answer a query $F \subseteq S \cup U$ with $|F| \leq f$, we consider $W = F-S \subseteq U$, and note that $S \cup F = S \cup W$.
    So, in light of~\Cref{lem:US-lemma}, we do the following:
    If $F \supseteq S$ and $G-(S \cup W)$ is disconnected, return \emph{``cut''};
    If $F \not \supseteq S$ and the array of $W$ contains some subset $F' \subseteq F$, return \emph{``cut''};
    Otherwise, return \emph{``fail''}.
    The soundness and completeness guarantees follow immediately from~\Cref{lem:US-lemma}.
    The query time is dominated by the
    second part, where we need to go over all subsets of $F$, and check if they appear in the sorted array corresponding to $W$.
    Each such check takes $O(f \log n)$ time, as we do a binary search in an array of length $\leq n$, where each element is represented by $f$ (or less) words.
    So, the query time is $O(2^{|F|} \cdot f \log n)$.
    The improvement in query time when $G$ is $f$-connected is because in this case, the array of $W$ cannot contain a neighbor set $N(C)$ with $|N(C)| < f$; indeed, in such a case $N(C)$ would be a cut in $G$ (separating $C$ from some vertex in $F - N(C)$), contradicting that $G$ is $f$-connected.
    Thus, when $G$ is $f$-connected, we do not need to go over all subsets of $F$, but rather just check if $F$ itself appears in the array of $W$, which takes $O(f \log n)$ time.
\end{proof}

%% file: left-right-graphs.tex
\subsection{Left and Right Graphs}\label{sec:left-right-graphs}

In this section we formally define the notions of left and right graphs and prove their ``cut respecting'' properties.
We remark that while~\cite{NSYArxiv23} were interested only in \emph{minimum} vertex cuts, we are interested in \emph{any} vertex cut of size $\leq f$.
For this reason, we introduce slight adaptations to their original construction that fit our needs, and present it in a self-contained manner.

Let $G = (V,E)$ be a connected graph with terminals $T \subseteq V$, let $f \geq 1$ be an integer parameter, and let $(L,S,R)$ be some vertex $T$-cut in $G$.

\begin{definition}[Left/Right Graphs]\label{def:left-right-graphs}
    The \emph{$f$-left graph} $G_L$ w.r.t.\ $G$, $T$ and $(L,S,R)$ is obtained from $G$ as follows:
    \begin{itemize}
        \item Order $R$ as a list where the terminals in $T \cap R$ appear first, and take $U_R$ as the first $f+1$ vertices in the list (or $U_R = R$ if $|R| \leq f+1$).
        \item Delete all vertices in $R$ other than $U_R$ (i.e., delete $R - U_R$).
        \item Connect $U_R$ as a clique, by adding all of the $U_R \times U_R$ edges that are missing in $G$.
        \item Connect $U_R$ and $S$ as a biclique, by adding all of the $U_R \times S$ edges that are missing in $G$.
    \end{itemize}
    The $f$-right graph $G_R$ is defined in symmetrically. 
\end{definition}

\begin{lemma}[Completeness of Left/Right Graphs]\label{lem:cut-in-G-to-left-right-graphs}
    Let $F \subseteq V$ with $|F| \leq f$.
    Suppose that in $G$, $F$ separates $T$ but does not separate $S$.
    Then either
    $F \cap V(G_L)$ separates $T \cap V(G_L)$ in $G_L$, or $F \cap V(G_R)$ separates $T \cap V(G_R)$ in $G_R$.
\end{lemma}
\begin{proof}
    Let $(A,F,B)$ be a vertex $T$-cut in $G$.
    As $F$ does not separate $S$, we may assume $B \cap S = \emptyset$ (otherwise $A \cap S = \emptyset$, so we can swap $A$ and $B$).
    Choose some terminal $t^* \in B \cap T$.
    Then $t^* \notin S$, hence $t^* \in L \cup R$.
    We assume $t^* \in L$ (otherwise, swap $L$ and $R$ in the following arguments). 
    
    Consider the partition $(A \cup R,F-R,B-R)$ of $V$, obtained from $(A,F,B)$ by moving all the $R$-vertices 
    into the first set.
    Since $B \cap S = \emptyset$, we have $B-R = B \cap (L \cup S) = B \cap L$.
    There are no edges between $A \cup R$ and $B \cap L$ (as no edges go between $A$ and $B$, nor between $R$ and $L$).
    Thus, $(A \cup R, F-R, B\cap L)$ is a cut in $G$.
    We ``project'' it onto $G_L$ by replacing $R$ with $U_R$; 
    We get the partition $(A' = (A-R) \cup U_R, F - U_R , B \cap L)$ of $V(G_L)$.
    Observe that still, there are no edges between the projected sides $A'$ and $B \cap L$,
    as all new edges in $G_L$ have two endpoints inside $U_R \cup S \subseteq R \cup S$, so they cannot have an endpoint in the side $B \cap L \subseteq L$.

    We assert that the side $A'$ must contain some terminal $t \in T$ such that $t \notin F$.
    If $U_R \cap T \not \subseteq F$, we can just choose some $t \in (U_R \cap T) - F$.
    Otherwise, $|U_R \cap T| \leq |F| \leq f$, so by definition of $U_R$ it must be that $T \cap R = U_R \cap T \subseteq F$.
    In this case we take some $t \in A \cap T$, which exists as $(A,F,B)$ is a $T$-cut.
    Then $t \notin F$, and hence also $t \notin R$ (as $R \cap T \subseteq F$), so $t \in A-R \subseteq A'$.

    We also have the terminal $t^* \in T$ in the other side $B \cap T$, and $t^* \notin F$.
    Thus, in $G_L$, $F-U_R$ separates the two terminals $t, t^* \in T \cap V(G_L)$, where $t, t^* \notin F$.
    As $F - U_R \subseteq F \cap V(G_L)$, $t$ and $t^*$ are also separated by $F \cap V(G_L)$.
    Thus, $F \cap V(G_L)$ separates $T \cap V(G_L)$ in $G_L$.
\end{proof}

\begin{lemma}[Soundness of Left/Right Graphs]\label{lem:cut-in-left-right-is-cut-in-G}
    Let $F \subseteq V(G_L)$ (resp., $F \subseteq V(G_R)$) with $|F| \leq f$.
    Suppose $F$ separates two vertices $x,y$ in $G_L$ (resp., $G_R$), 
    Then $F$ separates $x,y$ also in $G$.
\end{lemma}
\begin{proof}
    We prove the lemma for $G_L$ ($G_R$ is symmetric). 
    We split into two cases.
    
    Case 1: $U_R \subseteq F$.
    Then $|U_R| \leq |F| \leq f$. So, by definition of $U_R$, we have $U_R = R$.
    Thus, in this case, $G$ is a subgraph of $G_L$, so $F$ must also separate $x$ and $y$ in $G$.

    Case 2: $U_R \not \subseteq F$.
    Seeking contradiction, suppose there is a path $P$ from $x$ to $y$ in $G - F$.
    Since $G_L - F$ contains $G[V(G_L)] - F$ as a subgraph, any contiguous subpath of $P$ that only uses vertices in $V(G_L)$ is also present in $G_L - F$.
    Now, consider some maximal contiguous subpath $Q$ of $P$ such that $V(Q) \cap V(G_L) = \emptyset$.
    Let $u$ and $v$ be the vertices appearing right before and after $Q$ on $P$.
    As both $u$ and $v$ have $G$-neighbors in $V(Q) \subseteq V - V(G_L) \subseteq R$, it must be that $u,v \in S \cup R$.
    Also, by the maximality of $Q$, we have $u,v \in V(G_L)$, so in fact, $u,v \in S \cup U_R$.
    Thus, in $G_L - F$, both $u$ and $v$ are neighbors of some $r \in U_R - F$ (which exist as $U_R \not \subseteq F$).
    Therefore, we can just ``shortcut'' the segment $Q$ and replace it by walking from $u$ to $v$ through $r$ in $G_L - F$.
    This shortcutting procedure gives us a path from $x$ to $y$ in $G_L - F$, which is a contradiction.
\end{proof}

We now give another structural lemma on $f$-left/right graphs, which is aimed to tackle a slightly technical issue discussed in more detail in the next~\Cref{sec:left-right-tree} (\Cref{remark:right-child}).

\begin{lemma}\label{lem:stepchild-lemma}
    Let $F \subseteq V(G_R)$ with $|F| \leq f$.
    Suppose that in $G_R$, $F$ separates $T \cap V(G_R)$, but does not separate $S$ and does not separate $T \cap R$.
    Let $U_S$ be an arbitrary set of $f+1$ terminals in $T\cap S$ (or all of $T\cap S$ if $|T \cap S| \leq f+1$).
    Then $F$ separates $U_L \cup U_R \cup U_S$ in $G_R$. 
\end{lemma}
\begin{proof}
    Seeking contradiction, assume that all vertices $(U_L \cup U_R \cup U_S) - F$ lie in the same connected component $C$ of $G_R - F$.
    We show that every $t \in T \cap V(G_R) - F$ must also lie in $C$, contradicting that $F$ separates $T \cap V(G_R)$.
    If $t \in U_L$ we are done.
    If $t \in S$, then it cannot be that $U_S \subseteq F$, because then $|U_S| \leq |F| \leq f$, hence $t \in T \cap S = U_S \subseteq F$, but we now that $t \notin F$.
    So, we can choose some $t' \in U_S - F$.
    As $F$ does not separate $S$, $t$ must be in the same component of $t'$, which is $C$.
    If $t \in R$, we can find some $t' \in T \cap U_R - F$ by a similar argument to the previous case.
    As $F$ does not separate $T \cap S$, we again get that $t$ is in $C$ together with $t'$.
\end{proof}

Finally, we observe the following nice property of the $f$-left/right graphs:
\begin{observation}\label{obs:arboricity-left-right}
    If $G$ has arboricity $\alpha$, then $G_L$ and $G_R$ both have arboricity at most $\alpha + f + 1$.
\end{observation}
\begin{proof}
    In $G_L$, the edges internal to $L \cup S$ are all original to $G$, so they can be covered by $\alpha$ forests (the original $\alpha$ forests covering $G$, induced on $L \cup S$).
    The remaining edges all touch $U_R$, so they can be covered by $|U_R| \leq f+1$ stars centered at the $U_R$-vertices.
    $G_R$ is symmetric.
\end{proof}

%% file: LRtree.tex
\subsection{The Left-Right Decomposition Tree}\label{sec:left-right-tree}

In this section we describe the decomposition of the original instance $(G,T)$ by splitting into left and right graphs using a ``good'' $T$-cut $(L,S,R)$, and continuing recursively until we get terminal expanders or instances with few terminals.
The cut $(L,S,R)$ is produced by the following powerful lemma of~\cite{LongS22full}:
\begin{lemma}[Lemma 4.6 in~\cite{LongS22full}]\label{lem:balanced-or-expander}
    Let $G = (V,E)$ be a connected graph with $|V| = n$ vertices and $|E| = m$ edges.
    Let $T \subseteq V$ be a set of terminals in $G$.
    Let $0 < \epsilon \leq 1$ and $1 \leq r \leq \lfloor \log_{20} n \rfloor$ be parameters.
    There is a deterministic algorithm that computes a $T$-cut $(L,S,R)$ in $G$ (or possibly $L=S=\emptyset$) such that $|S| \leq \epsilon |T \cap (L \cup S)|$, which further satisfies either
    \begin{itemize}
        \item \emph{(``balanced terminal cut'')} $|T \cap (L \cup S)|, |T \cap (R \cup S)| \geq \frac{1}{3} |T|$, or
        \item \emph{(``expander'')} $|T \cap R| \geq \frac{1}{2} |T|$ and $G[R]$ is a $(T\cap R, \phi)$-expander for some $\phi \geq \epsilon / (\log n)^{O(r^5)}$.
    \end{itemize}
    The running time is $O(m^{1+o(1)+O(1/r)} \cdot (\log m)^{O(r^4)} / \phi)$.
\end{lemma}

\paragraph{Construction of the $f$-LR Tree $\T$.}
Let $G = (V,E)$ be a graph with $|V| = n$, $|E| = m$.
Let $T \subseteq V$ be a terminal set in $G$.
Formally, the \emph{$f$-left-right decomposition tree} (or \emph{$f$-LR tree} for short) with respect to the instance $(G,T)$ is a (virtual) rooted tree $\mathcal{T}=(\mathcal{V},\mathcal{E})$.
For ease notation, we will use the letter $q$ (sometimes with subscripts such as $q_l, q_r$) to denote nodes of $\T$, to distinguish them from the vertices of $G$ that are usually denoted by $x,y$ or $u,v,w$.
Each node $q \in \mathcal{V}$ is associated with a pair of graph and terminal set $(G_q,T_q)$, such that $T_q \subseteq V(G_q) \subseteq V(G)$.
If $q$ is an internal (i.e., non-leaf) node in $\T$, it is additionally associated with a vertex cut $(L_q, S_q, R_q)$ in $G_q$ (where possibly $L_q = S_q = \emptyset$).
The tree $\T$ is constructed by a recursive algorithm, described next, that makes calls to~\Cref{lem:balanced-or-expander}.
We define
\begin{equation}\label{eq:eps-balanced-or-exp-lemma}
    \epsilon := \frac{1}{c \log |T|} \text{ for large enough constant $c > 1$,} \quad \text{and } r := O(1).
\end{equation}
These fixed values of $\epsilon$ and $r$ are used in every invocation of~\Cref{lem:balanced-or-expander}. 
Thus, the corresponding expansion parameter from this lemma is always at least
\begin{equation}\label{eq:expansion}
    \phi := \frac{\epsilon}{(\log n)^{O(r^5)}} = \frac{1}{\poly \log n}
\end{equation}
Note that for any constant $\delta > 0$ we can take $r$ to be a sufficiently large constant (depending only on $\delta$) to get running time $m^{1+\delta}$ in~\Cref{lem:balanced-or-expander}.

We now describe algorithm for constructing $\T$.
We initialize $\T$ to have only a root node associated with $(G,T)$.
We then apply a recursive procedure from this root.
In general, each recursive call is given a node $q$ which is a leaf in the current state of $\T$, and either decides to make $q$ a permanent leaf (and end the current recursive branch), or creates new children nodes for $q$ and applies recursive calls on (some or all of) them.

The recursive call on $q$ is implemented as follows.
First, if $|T_q| \leq (f+1)/\epsilon$, then $q$ is set to be a permanent leaf in $\T$, and the call is terminated.
Otherwise, we apply~\Cref{lem:balanced-or-expander} on $G_q$ and get the cut $(L_q, S_q, R_q)$.
We then construct the $f$-left and $f$-right graphs of $(G_q, T_q)$ w.r.t.\ $(L_q, S_q, R_q)$, denoted $G_{L_q}$ and $G_{R_q}$ respectively.
We also denote by $U_{L_q}$ and $U_{R_q}$ the sets of~\Cref{def:left-right-graphs} used to construct these graphs.
Additionally, let $U_{S_q}$ be an arbitrary set of $f+1$ terminals from $T_q \cap S_q$ (or all of $T_q \cap S_q$ if $|T_q \cap S_q| \leq f+1$).
The algorithm now proceeds according to the two cases of~\Cref{lem:balanced-or-expander}:

\begin{description}
    \item[(Balanced)]
    If $|T_q \cap (L_q \cup S_q)|, |T_q \cap (R_q \cup S_q)| \geq \frac{1}{3} |T_q|$:
    
    We then create two children for $q$:
    \begin{itemize}
        \item A \emph{left child} $q_l$ associated with $(G_{L_q}, T_q \cap V(G_{L_q}) )$
        \item A \emph{right child} $q_r$ associated with $(G_{R_q}, T_q \cap V(G_{R_q}) )$.
    \end{itemize}
    We then apply the algorithm recursively on both $q_l$ and $q_r$.

    \item[(Expander)]
    Else, we know that $|T_q \cap R_q| \geq \frac{1}{2}|T_q| $ and that $G_q[R_q]$ is a $(T_q \cap R_q, \phi)$-expander.

    In this case, we create three children for $q$:
    \begin{itemize}
        \item A \emph{left child} $q_l$ associated with $(G_{L_q}, T_q \cap V(G_{L_q}) )$
        \item A \emph{right child} $q_r$ associated with $(G_{R_q}, T_q \cap R_q )$.
        \item An additional \emph{stepchild} $q_s$ associated with $(G_{R_q}, U_{L_q} \cup U_{R_q} \cup U_{S_q})$
    \end{itemize}
    We then apply the algorithm recursively only on $q_l$
    ($q_r,q_s$ are permanent leaf nodes).
\end{description}

\begin{remark}\label{remark:right-child}
The right child $q_r$ is associated with $G_{R_q}$ in both cases, but its terminal set is defined slightly differently. 
In the balanced case, we simply take all terminals of $T_q$ found in $G_{R_q}$;
in the expander case, we only take the terminals from $R_q$.
This subtle distinction is to ensure that in the latter case, the instance at $q_r$ is indeed a terminal expander, as $G_q[R_q]$ is a $(T_q \cap R_q, \phi)$-expander, and $G_{R_q}$ is a \emph{supergraph} of $G_q[R_q]$ (while $G_{R_q}$ might not be an expander w.r.t.\ \emph{all} the $T_q$-terminals in it).
The point of the additional ``stepchild'' $q_s$ is essentially to account for this loss of terminals in the right child, by using~\Cref{lem:stepchild-lemma}.
\end{remark}

Finally, we define the set $S^*$ for~\Cref{thm:cut-detectors} as the union of all separators $S_q$ from the cuts associated with internal nodes of $\T$:
\begin{equation}\label{eq:terminalLRtree}
    S^* := \bigcup \big\{ S_q \mid \text{$q \in \mathcal{V}$ is an internal node of $\T$} \big\} .
\end{equation}

\paragraph{Analysis.}
We now prove the key properties of the $f$-LR tree $\T$:

\begin{lemma}\label{lem:fLR-Tree}
    Let $\mathcal{T} = (\mathcal{V}, \mathcal{E})$ be the LR-tree for $(G,T)$ and let $S^*$ be the union of separators associated with nodes of $\T$ as defined in~\Cref{eq:terminalLRtree}.
    Then all of the following properties hold.
     \begin{enumerate}
        \item \emph{(Expander/small leaf-nodes)} Let $q$ be a leaf node in $\T$ associated with $(G_q, T_q)$.
        Then either (i) $G_q$ is a $(T_q, \phi)$-expander, or (ii) $|T_q| = O(f/\epsilon)$.
        (See~\Cref{eq:eps-balanced-or-exp-lemma,eq:expansion} for $\epsilon$ and $\phi$.)
        \label{prop:expander-or-small-leaves}
    
        \item \emph{(Logarithmic depth)} The depth of $\T$ is $d = O(\log |T|)$.
        \label{prop:log-depth}
    
        \item \emph{(Terminal reduction)} $|S^*| \leq \frac{1}{2} |T|$.
        \label{prop:terminal-reduction}
    
        \item \emph{(Near-linear space)} $\sum_{q \in \mathcal{V}} |V(G_q)| = O(nd)$ and $\sum_{q \in \mathcal{V}} |E(G_q)| = O(fnd^2)$.
        \label{prop:linear-space}
    \end{enumerate}
\end{lemma}

The rest of this section is devoted to proving the above~\Cref{lem:fLR-Tree}.
Observe that Property~\ref{prop:expander-or-small-leaves} regarding the leaves of $\T$ follows immediately from the description of the algorithm and the discussion in~\Cref{remark:right-child}.
(Note that a ``stepchild'' of the form $q_s$ is a leaf node associated with at most $3(f+1) \leq O(f/\epsilon)$ terminals, as $|U_{L_q}|,|U_{R_q}|,|U_{S_q}| \leq f+1$.)

For the remaining properties, we need the following two claims.
In both, $q$ refers to some internal node in $\T$.
Hence, $q$ has left and right children $q_l$ and $q_r$ (and possibly also a stepchild $q_s$ which is a leaf in $\T$).
Recall that because $q$ is not a leaf, we must have $|T_q| > (f+1)/\epsilon$, so $f+1 < \epsilon |T_q|$.
Also, recall that $|S_q| \leq \epsilon |T_q|$ by~\Cref{lem:balanced-or-expander}, and that $|U_{L_q}|,|U_{R_q}|\leq f+1$ by~\Cref{def:left-right-graphs}.
These facts will serve us in both proofs.
\begin{claim}\label{clm:terminals-shrinking}
    Let $q$ be an internal node in $\T$.
    Then it holds that $|T_{q_l}| \leq 0.9 |T_q|$.
    Additionally, if $q_r$ is an internal node in $\T$, then also $|T_{q_r}| \leq 0.9 |T_q|$.
\end{claim}
\begin{proof}
    Suppose first that the balanced case occurred in the recursive call initiated at $q$.
    Then we have $|T_q \cap (R_q \cup S_q)| \geq \frac{1}{3} |T_q|$, and hence $|T_q \cap L_q| \leq \frac{2}{3} |T_q|$.
    We thus get
    \[
    |T_{q_l}| = |T_q \cap V(G_{L_q})| = |T_q \cap L_q| + |T_q \cap S_q| + |T_q \cap U_{R_q}| \leq \left( \frac{2}{3} + 2\epsilon \right) |T_q| \leq 0.9 |T_q|.
    \]
    Additionally, in the balanced case we also have $|T_q \cap (L_q \cup S_q)| \geq \frac{1}{3} |T_q|$, so a symmetric argument proves that also $|T_{q_r}| \leq 0.9 |T_q|$.

    It the expander case $q_r$ is a leaf, so we only have to care about $|T_{q_l}|$.
    But in this case we have $|T_q \cap R_q| \geq \frac{1}{2} |T_q|$, so in particular $|T_q \cap (R_q \cup S_q)| \geq \frac{1}{3} |T_q|$, and the same calculation from the balanced case goes through.
\end{proof}

\begin{claim}\label{clm:space-induction-in-LR-tree}
    Let $q$ be an internal node in $\T$.
    Then the following inequalities hold:
    \begin{equation}\label{eq:vertices-recursion}
        |V(G_{q_l})| + |V(G_{q_r})| \leq (1+3\epsilon) |V(G_q)|,
    \end{equation}
    \begin{equation}\label{eq:terminals-recursion}
         |T_{q_l}| + |T_{q_r}| \leq (1+3\epsilon) |T_q|.
    \end{equation}
\end{claim}
\begin{proof}
    Recall that $G_{q_l} = G_{L_q}$ and $G_{q_r} = G_{R_q}$.
    So, $|V(G_{q_l})| + |V(G_{q_r})|$ counts twice the vertices in $S_q \cup U_{L_q} \cup U_{R_q}$, and counts once every other vertex in $G_q$.
    We get
    \[
    |V(G_{q_l})| + |V(G_{q_r})|
    \leq |V(G_q)| + |S_q| + 2(f+1) \leq |V(G_q)| + 3\epsilon|T_q| \leq (1+3\epsilon) |V(G_q)|
    \]
    which proves~\Cref{eq:vertices-recursion}.
    The proof of~\Cref{eq:terminals-recursion} is omitted: it is essentially identical, only now counting the terminals of $T_q$ that are found in the $f$-left and $f$-right graphs $G_{L_q}$ and $G_{R_q}$.
\end{proof}

We are now ready to show the remaining properties of~\Cref{lem:fLR-Tree} (Properties~\ref{prop:log-depth}, \ref{prop:terminal-reduction} and \ref{prop:linear-space}).
We start with Property~\ref{prop:log-depth}:
that $\T$ has depth $d = O(\log |T|)$ follows from~\Cref{clm:terminals-shrinking}, implying that when walking down from the root of $\T$ to its deepest leaf, the number of terminals associated with the current node, which starts at $|T|$, shrinks by a $0.9$ factor in every step (except maybe the last).

Next, we prove Property~\ref{prop:terminal-reduction}, bounding the size of $|S^*|$.
Let $\mathcal{V}_i$ denote the set of nodes with depth $i$ in $\T$.
By inductively using~\Cref{eq:terminals-recursion}, we obtain that for every $i \in \{0,1,\dots,d\}$,
\[
\sum_{\text{internal $q \in \mathcal{V}_i$}} |T_q| \leq (1+3\epsilon)^i |T|
\]
(here, we used the fact that stepchildren nodes cannot be internal in $\T$).
Recall that for every internal $q \in \mathcal{V}$ we have $|S_q| \leq \epsilon |T_q|$ by~\Cref{lem:balanced-or-expander},
so we obtain
\[
|S^*| 
\leq \sum_{\text{internal $q \in \mathcal{V}$}} |S_q| 
\leq \sum_{i=1}^d \sum_{\text{internal $q \in \mathcal{V}_i$}} \epsilon |T_q|
\leq \epsilon \Big( \sum_{i=1}^d (1 + 3\epsilon)^i \Big)|T|
= \frac{(1+3\epsilon)^{d+1}-1}{3} |T|
\leq \frac{1}{2} |T|,
\]
where in the last inequality, we use that $d = O(\log |T|)$ (as proved in Property~\ref{prop:log-depth}), so we can choose the constant $c$ in~\Cref{eq:eps-balanced-or-exp-lemma} to be large enough so as to make $(1+3\epsilon)^{d+1} \leq 2.5$.

Finally, we show Property~\ref{prop:linear-space}.
Let $\widehat{\mathcal{V}}_i$ denote the subset of $\mathcal{V}_i$ that contains all the nodes with depth $i$ in $\T$ that are not stepchildren.
By inductively using~\Cref{eq:vertices-recursion}, we obtain that for every $i \in \{0,1,\dots,d\}$,
\[
\sum_{q \in \widehat{\mathcal{V}}_i} |V(G_q)| \leq (1+3\epsilon)^i n \leq O(n),
\]
where the last inequality follows from $i \leq d = O(\log |T|)$ and the choice of $\epsilon$ in~\Cref{eq:eps-balanced-or-exp-lemma}.
Now, recall that each stepchild has the same number of vertices as its sibling right child (as they are both associated with the same graph, but with different terminals).
Therefore, we get
\[
\sum_{q \in \mathcal{V}} |V(G_q)|
= \sum_{i=1}^d \sum_{q \in \mathcal{V}_i } |V(G_q)|
\leq \sum_{i=1}^d  \Big( 2 \sum_{q \in \widehat{\mathcal{V}}_i } |V(G_q)| \Big) 
\leq O(nd).
\]
Next, by inductively applying~\Cref{obs:arboricity-left-right}, and recalling that the root of $\T$ is associated $G$, whose arboricity is at most $f+1$, we see that if $q \in \mathcal{V}_i$ then $G_q$ has arboricity at most $(i+1) (f+1)$, and hence at most $|E(G_q)| \leq (i+1)(f+1)|V(G_q)| \leq O(df) \cdot |V(G_q)|$.
So, summing over all nodes in $\T$, we get
\[
\sum_{q \in \mathcal{V}} |E(G_q)| = O(df) \cdot \sum_{q \in \mathcal{V}} |V(G_q)| = O(df) \cdot O(nd) = O(fn d^2).
\]

This concludes the proof of~\Cref{lem:fLR-Tree}.

%% file: data-structure.tex
\subsection{The Terminal Cut Detector $\D$ of~\Cref{thm:cut-detectors}}\label{sec:construction-and-query-of-D}

Let $G$, $T$ and $f$ be the same as in the previous~\Cref{sec:left-right-tree}.
This last section discussed how to find the set $S^*$ for~\Cref{thm:cut-detectors} by constructing the $f$-LR tree $\T$ constructed for $(G,T)$.
In this section, we complete the proof of~\Cref{thm:cut-detectors} by providing the construction of the $(f,T,S^*)$-cut detector $\D$ by using the $f$-LR tree $\T$ and the specialized cut detectors from~\Cref{sec:special-cut-detectors}.

\paragraph{Constructing $\D$.}
The construction starts by computing the $f$-LR tree $\T$ for $(G,T)$ and the corresponding new terminal set $S^*$ (defined in~\Cref{eq:terminalLRtree}).

The next step is augmenting each leaf node $q$ of $\T$, associated with $(G_q, T_q)$, with a corresponding specialized $(f,T_q,\emptyset)$-cut detector.
By~\Cref{lem:fLR-Tree}(\ref{prop:expander-or-small-leaves}), either $|T_q| = O(f/\epsilon)$, or $G_q$ is a $(T_q, \phi)$-expander (with $\phi$ as in~\Cref{eq:expansion}).
So, if the former option holds, we can use $\FEOracle(G_q,T_q,f)$; otherwise we use $\TEOracle(G_q, T_q,\phi,f)$.

Next, we augment the internal nodes of $\T$ with specialized US cut detectors.
Let $q$ be an internal node in $T$, associated with $(G_q, T_q)$ and with the vertex cut $(L_q, S_q, R_q)$ in $G$.
Let $G_{L_q}$ and $G_{R_q}$ be the corresponding $f$-left and $f$-right graphs, and $U_{L_q}$ and $U_{R_q}$ as in~\Cref{def:left-right-graphs}.
The node $q$ will be augmented with two US cut detectors:
$\USOracle(G_{L_q}, U_{L_q} \cup U_{R_q}, S_q, f)$ and $\USOracle(G_{R_q}, U_{L_q} \cup U_{R_q}, S_q, f)$.
Note that both have the same ``S-set'' $S_q$ and the same ``U-set'' $U_{L_q} \cup U_{R_q}$, but they are constructed with respect to the different graphs $G_{L_q}$ and $G_{R_q}$.

When we finish augmenting all the nodes in $\T$ as described above, we no longer need to explicitly store the graphs $G_q$ in each node $q$.
Instead, we just store $V(G_q)$, and (in case $q$ is internal) its partition into the cut $(L_q, S_q, R_q)$ and the sets $U_{L_q}$ and $U_{R_q}$.
We keep storing the associated terminal set $T_q$.
Each of these vertex subsets is stored in some data structure supporting membership queries in $\tilde{O}(1)$ time.%
\footnote{We can use perfect hashing to get constant query time for membership queries. However, as we already incur other logarithmic factors, a balanced binary search tree or a sorted array will suffice for us.}
This concludes the construction of $\D$.

\paragraph{Query Algorithm of $\D$.}
We now describe how the $(f,T,S^*)$-cut detector $\D$ answers a given query $F \subseteq V$ with $|F| \leq f$.
The query algorithm explores $\T$ in a recursive manner, which is initialized from the root of $\T$.
In each visited node $q$, the recursive call associated with $q$ returns either \emph{``cut''} or \emph{``fail''}, and the final answer is the value returned by the initial call to the root.

We denote $F_q = F \cap V(G_q)$,
and when $q$ is an internal node in $\T$, $F_l = F \cap V(G_{L_q})$ and $F_r = F \cap V(G_{R_q})$.
Note that $F_{q_l} = F_l$, $F_{q_r} = F_r$, and $F_{q_s} = F_r$ when $q_s$ exists.

The recursion is implemented as follows: %
\begin{mdframed}[style=MyFrame]
\begin{description}
    \item[\textbf{(``Leaf'')}] If $q$ is a leaf:
    Then we query the $(f,T_q,\emptyset)$-cut detector of $G_q$ which is found in $q$
    (which is either $\FEOracle(G_q,T_q,f)$ or $\TEOracle(G_q,T_q,\phi,f)$) with $F_q$,
    and return the answer obtained from this query.

    \item[\textbf{(``Trim Right'')}] Else, if $F_q \cap R_q \subseteq U_{R_q}$:
    Then we have $F_r \subseteq S_q \cup U_{L_q} \cup U_{R_q}$.
    Hence, we can query $\USOracle(G_{R_q}, U_{L_q} \cup U_{R_q}, S_q, f)$ with $F_r$.
    \begin{itemize}
        \item If this query returns \emph{``cut''}, we return \emph{``cut''}.
        \item Otherwise, we recurse only on the left child $q_l$.
        If this recursive call returned \emph{``cut''}, we return \emph{``cut''}; otherwise we return \emph{``fail''}.
    \end{itemize}
    
    \item[\textbf{(``Trim Left'')}] Else, if $F_q \cap L_q \subseteq U_{L_q}$:
    Then in a similar fashion to the previous case, 
    we can query $\USOracle(G_{L_q}, U_{L_q} \cup U_{R_q}, S_q, f)$ with $F_l$.
    
    \begin{itemize}
        \item If this query returns \emph{``cut''}, we return \emph{``cut''}.
        \item Otherwise, we recurse on the right child $q_r$ and on the stepchild $q_s$ (if exists).
        If at least one recursive call returned \emph{``cut''}, we return \emph{``cut''}; otherwise we return \emph{``fail''}.
    \end{itemize}

    \item[\textbf{(``Branch'')}] Else:
    We recurse on all children of $q$.
    If at least one recursive call returned \emph{``cut''}, we return \emph{``cut''}; otherwise we return \emph{``fail''}.
\end{description}
\end{mdframed}
We denote by $\T(F)$ the \emph{query tree of $F$}, induced by $\T$ on the nodes visited during the query;
note that $\T(F)$ is connected and contains the root of $\T$.

\paragraph{Correctness.}
We will show the following by induction from the leaves upwards on $\T(F)$:
\begin{lemma}\label{lem:corretness-induction}
    The answer returned from a recursive call invoked on a node $q$ in $\T(F)$ satisfies:
    \begin{itemize}
        \item (Soundness) If the answer is \emph{``cut''}, then $F_q$ is a cut in $G_q$.
        \item (Completeness) If $F_q$ separates $T_q$ but not $S^* \cap V(G_q)$ in $G_q$, then \emph{``cut''} is returned.
\end{itemize}
\end{lemma}
The correctness of the $(f,T,S^*)$-cut detector $\D$ immediately follows by taking $q$ as the root of $\T$ in~\Cref{lem:corretness-induction}.
The argument hinges on the completeness and soundness properties of left and right graphs proved in~\Cref{lem:cut-in-G-to-left-right-graphs} and~\Cref{lem:cut-in-left-right-is-cut-in-G}.
We now give the full proof.

\begin{proof}[Proof of~\Cref{lem:corretness-induction}]
    The ``Leaf'' case serves as the base case for the induction:
    then, we return the answer of an $(f,T_q, \emptyset)$-cut detector in $G_q$ on the query $F_q$, so both the soundness and completeness properties clearly hold.
    From now on, assume that $q$ is an internal node.
    
    We start by showing the soundness property.
    Observe that a \emph{``cut''} answer from $q$ can happen only if \emph{``cut''} was returned from some recursive call invoked from a child of $q$, which is associated with $G_{L_q}$ or $G_{R_q}$, or if some US cut detector built for $G_{L_q}$ or $G_{R_q}$ returned \emph{``cut''}.
    Such \emph{``cut''} answers are sound, either by induction hypothesis or by the correctness of US cut detectors from~\Cref{lem:US-data-structure}.
    Namely, \emph{``cut''} can be returned at $q$ only in case $F_l$ is a cut in $G_{L_q}$ or $F_r$ is a cut in $G_{R_q}$, and~\Cref{lem:cut-in-left-right-is-cut-in-G} implies that $F_q$ is indeed cut in $G_q$ in this case.
    
    We now show the completeness property, so assume $F_q$ separates $T_q$ but not $S^* \cap V(G_q)$ in $G_q$.
    First, by~\Cref{lem:cut-in-G-to-left-right-graphs}, we have that either (i) $F_l$ separates $T_q \cap V(G_{L_q})$ in $G_{L_q}$, or (ii) $F_r$ separates $T_q \cap V(G_{R_q})$ in $G_{R_q}$.
    These cases are almost symmetric, except (ii) introduces a slight complication due to the possible existence of a stepchild $q_s$.
    So, we will consider both cases:
    \begin{itemize}
        \item[(i)] The contrapositive of~\Cref{lem:cut-in-left-right-is-cut-in-G} implies that $F_l$ does not separate $S^* \cap V(G_{L_q})$ in $G_{L_q}$. In particular, by~\Cref{eq:terminalLRtree}, $F_l$ does not separate $S_q$ in $G_{L_q}$.
        In the ``Trim Right'' case, we query the $(f,V(G_{L_q}),S_q)$-cut detector with $F_l$, so it must return \emph{``cut''}.
        In the remaining cases, a recursive call is invoked from $q_l$ (or \emph{``cut''} is already returned and we are done). This call must return \emph{``cut''} by the induction hypothesis and the discussion above.
        So, in any case, the call from $q$ also returns \emph{``cut''}.

        \item[(ii)] By symmetric arguments as (i), $F_r$ does not separate $S^* \cap V(G_{R_q})$ nor $S_q$ in $G_{R_q}$.
        Further, the ``Trim Left'' case must return \emph{``cut''}.
        In the remaining cases, recursive calls are invoked from $q_r$ and from the stepchild $q_s$ if it exists (or \emph{``cut''} is already returned and we are done).
        So, it suffices to show that one of these calls must return \emph{``cut''}.
        If the stepchild $q_s$ doesn't exist, the argument is symmetric to case (i), so assume $q_s$ exists.
        If $F_r$ separates $T_q \cap R_q$ in $G_{R_q}$, then the call on $q_r$ must return \emph{``cut''} by induction hypothesis.        
        Otherwise, by~\Cref{lem:stepchild-lemma}, $F_r$ separates $U_{L_q} \cup U_{S_q} \cup U_{R_q}$ in $G_{R_q}$, so the call on $q_s$ must return \emph{``cut''} by induction hypothesis.
        In any case, this means we return \emph{``cut''}.
    \end{itemize}
    The proof of~\Cref{lem:corretness-induction} is concluded.
\end{proof}

\paragraph{Query Time.}
We now analyze the query time, hinging on the following lemma:

\begin{lemma}\label{lem:branch-nodes-in-query}
    Consider a subtree of the query tree $\T(F)$ rooted at some node $q$, and let $x \geq 0$ be an integer.
    The number of ``Branch'' node $q'$ such that $|F_{q'}| = x$ in this subtree is at most $2^{|F_q|-x}$.
\end{lemma}
\begin{proof}
    By induction on the height of $\T_q (F)$.
    We split to cases according to $q$'s type:

        (''Leaf'')
        Trivial, as $q$'s subtree has $0$ ``Branch'' nodes (and $0 \leq 2^{|F_q|-x}$).

        (``Trim Right'')
        Then the ``Branch'' nodes in $q$'s subtree are exactly those from $q_l$'s subtree. So by induction, the number of such nodes $q'$ with $|F_{q'}| = x$ is at most $2^{|F_l|-x} \leq 2^{|F_q|-x}$.

        (``Trim Left'')
        Symmetric to the ``Trim Left'' case above.
        (Note that even if $q_s$ exists, it is a leaf, so its subtree cannot contribute to ``Branch'' nodes.)

        (``Branch'')
        Then we have $F_q \cap R_q \not \subseteq U_{R_q}$ (since $q$ is not ``Trim Right'').
        Hence (by~\Cref{def:left-right-graphs} of $f$-left graphs), at least one vertex in $F_q$ is missing from $F_l$ (this is some vertex in $F \cap R_q - U_{R_q}$), so $|F_l| \leq |F_q| - 1$.
        A symmetric argument gives that $|F_r| \leq |F_q|-1$.
        We now consider cases according to the value of $x$:
            When $x > |F_q|$, the subtree of $q$ cannot have a node with an associated query of size $x$ (and $0 \leq 2^{|F_q|-x}$).
            When $x = |F_q|$, the root $q$ is the only ``Branch'' node in its subtree with an associated query of size $x$ (and $1 = 2^{|F_q| - x}$).
            When $0 \leq x < |F_q|$, then ``Branch'' nodes with associated query of size $x$ in $q$'s subtree can come only from the subtrees of $q_l$ and $q_r$ (again, if $q_s$ exists, it is a leaf).
            By induction hypothesis, we obtain that these are at most $2^{|F_l|-x} + 2^{|F_r|-x} \leq 2 \cdot 2^{|F_q|-1-x} = 2^{|F_q|-x}$.
\end{proof}

Recall $d = O(\log |T|)$ is the depth of the $f$-LR tree $\T$ (see~\Cref{lem:fLR-Tree}).
We now get:
\begin{corollary}\label{cor:trim-and-total-nodes}
    The following hold for the query tree $\T(F)$:
    \begin{enumerate}
        \item $\T(F)$ has $O(d \cdot 2^{|F|-x})$ ``Trim'' nodes $q$ such that $|F_q| = x$, for every $0 \leq x \leq |F|$.\label{prop:trim-nodes}

        \item  $\T(F)$ has $O(d \cdot 2^{|F|})$ nodes overall.\label{prop:total-nodes}
    \end{enumerate}
\end{corollary}
\begin{proof}
    We let each `Trim'' node $q$ such that $|F_q| = x$ to choose as representative its nearest ancestor ``Branch'' node $b(q)$ in $\T(F)$ (if there is no such ancestor, let $b(q)$ be the root of $\T(F)$).
    Then each chosen $b(q)$ has $|F_{b(q)}| \geq x$, so by~\Cref{lem:branch-nodes-in-query} applied to the entire $\T(F)$, there are at most $1 + \sum_{y=x}^{|F|} 2^{|F|-y} = O(2^{|F|-x})$ different representatives.    
    Now, observe that ignoring the stepchildren (which are leaves) in $\T(F)$ gives a binary tree of depth $\leq d$, where the ``Branch'' nodes are exactly those with two children.
    Thus, we see that each representative node is chosen by at most $2d$ nodes, so item 1 follows.
    
    For item 2, note that by~\Cref{lem:branch-nodes-in-query}, the number of ``Branch'' nodes in $\T(F)$ is at most $\sum_{x=0}^{|F|} 2^{|F|-x} = O(2^{|F|})$.
    So the binary tree obtained by deleting the stepchildren has $O(d \cdot 2^{|F|})$ nodes, and adding back the stepchildren at most doubles this number.
\end{proof}

We are now ready to complete the query time analysis.
First, the time spent in a node $q$ of $\T(F)$ to identify 
the relevant case is $\tilde{O}(|F|)$, as we only need to check the membership of the vertices in $F$ in a constant number of vertex subsets.
The rest of the time spent in $q$ is by the query to the cut detector stored at $q$, which only happens in the ``Leaf'' or ``Trim'' cases. 
In the ``Leaf'' case, we either query a $\FEOracle$ structure constructed for $O(f/\epsilon) = O(f \log |T|)$ terminals, or $\TEOracle$ structure constructed for a terminal expander with expansion $\phi$; this takes at most $\tilde{O}(f^2 /\phi)$ time (by~\cref{lem:terminal-data-structure,thm:terminal-expander-DS}).
So, by~\Cref{cor:trim-and-total-nodes}(\ref{prop:total-nodes}), ignoring the time to query to the cut detectors in the ``Trim'' nodes, the query takes $\tilde{O}(2^f \cdot f^2/\phi)$ time.

It remains to analyze the ignored time for cut detector queries in ``Trim'' nodes.
In such a node $q$ with $|F_q| = x$, we make a query to a  $\USOracle$ structure which takes $\tilde{O}(2^x \cdot f)$ time by~\Cref{lem:US-data-structure}.
By~\Cref{cor:trim-and-total-nodes}(\ref{prop:trim-nodes}), summing over all ``Trim'' nodes, the total time spent is bounded by $\sum_{x=0}^{|F|} \tilde{O}(2^{|F|-x}) \cdot O(2^x \cdot f) = \tilde{O} (2^{|F|} \cdot f^2)$.

All in all, we get that the query $F$ is answered in $\tilde{O}(2^{|F|} \cdot f^2 / \phi)$ time.

\paragraph{Space.}
The space required for a specific node $q$ in $\T$ is dominated by the cut detector(s) with which $q$ is augmented, which is either $\FEOracle$, $\TEOracle$ or two $\USOracle$s.
It is readily verified that the latter case is the heaviest, requiring $\tilde{O}(f 2^{2f} |V(G_q)|)$ space by~\Cref{lem:US-data-structure} (recall that $|U_{L_q}|, |U_{R_q}| \leq f+1$).
Summing over all nodes $q$ in $\T$ and using~\Cref{lem:fLR-Tree}(\ref{prop:linear-space}), we get total space of $\tilde{O}(f 2^{2f} n)$ for $\D$.

\paragraph{Preprocessing Time.}
We first address the construction of the $f$-LR tree $\T$.
The computation at each node $q$ is dominated by applying~\Cref{lem:balanced-or-expander} on $G_q$. Thus, by the near-linear space of $\T$ in~\Cref{lem:fLR-Tree}(\ref{prop:linear-space}), the construction of $\T$ takes $\tilde{O}((fnd^2)^{1+o(1)+O(1/r)} \cdot (\log (fnd^2))^{O(r^4)} / \phi)$ time.

In the construction of $\D$ from $\T$, the bottleneck is augmenting each node $q$ in $\T$ with its cut detector(s).
If $q$ is a leaf, it is augmented either with $\TEOracle$ or with $\FEOracle$ constructed for $|E(G_q)|$, which takes $|E(G_q)|^{1+o(1)}/\phi + \tilde{O}(f|E(G_q)|)$ time by~\Cref{thm:terminal-expander-DS,lem:terminal-data-structure}.
If $q$ is a internal, it is augment with two $\USOracle$ constructed for its left and right graphs, which are $G_{q_r}$ and $G_{q_l}$, where the ``U-set'' has $\leq 2(f+1)$ vertices.
This takes $\tilde{O}(f 2^{2f} (|E(G_{q_l})| + |E(G_{q_r})|)$ time by~\Cref{lem:US-data-structure}.
Summing over all $q$ and using~\Cref{lem:fLR-Tree}(\ref{prop:linear-space}), we get $f^2 n^{1+o(1)}/\phi + \tilde{O}(f^2 2^{2f} n)$ time for constructing $\D$ from $\T$.

So, given any constant $\delta > 0$, we can set the parameter $r$ in~\Cref{eq:eps-balanced-or-exp-lemma} to be a large enough constant (depending only on $\delta$), and obtain the running time stated in~\Cref{thm:cut-detectors}, except for the extra $O(2^{2f})$ factor.
Alternatively, we can set $r = \Theta(\log \log n)$, so that the expansion becomes $\phi  = 1/n^{o(1)}$, resulting in the preprocessing improvement at cost of $n^{o(1)}$ factors in space and query time stated in~\Cref{thm:cut-detectors} (again, except for the extra $O(2^{2f})$ factor).

\medskip
This nearly concludes the proof of~\Cref{thm:cut-detectors}, except for $O(2^{2f})$ factors appearing the space of preprocessing time; \Cref{sec:space-improvment} explains how these can be shaved to get~\Cref{thm:cut-detectors}.
But, before we get to shaving this factors in the general case, we first provide the details on the $f$-connected case in the following~\Cref{sec:f-connected}.
(This is for ease of presentation purposes: the modifications for shaving the $O(2^{2f})$ factors in the general case are not needed in the $f$-connected case, so we prefer to defer them.)

%% file: f-connected-DS.tex
\subsection{Minimum Vertex Cut Oracles}\label{sec:f-connected} %

This section is devoted to \emph{minimum} vertex cuts, that when queried with any $F \subseteq V$, can determine if $F$ is a minimum vertex cut in $G$, hence providing a establishing Theorem \ref{thm:cut-theorem-fconnected}. 
This oracle is obtained by an improved variant of our $f$-vertex cut oracle when the given graph $G$ is $f$-connected.  The preprocessing first computes the vertex connectivity $f$ of $G$, and the space and query time of the resulting oracle also depend on $f$.

We do this by showing the relatively minor modifications to the general $f$-vertex cut oracle of~\Cref{thm:main-cut-theorem}, which utilize the additional $f$-connectivity promise to improve space and query time.
Specifically, we set our goal to improving the space and query time of the cut detector $\D$ from~\Cref{thm:cut-detectors} to $\tilde{O}(fn)$ and $\tilde{O}(f^2)$, respectively;
the corresponding complexities of the oracle are only larger by an $O(\log n)$ factor, as discussed in~\Cref{sec:overview-oracles}.

\paragraph{Improved Properties of Left/Right Graphs.}

Let $G$, $T$, $f$ and $(L,S,R)$ be as in~\Cref{sec:left-right-graphs}, with the corresponding $f$-left and $f$-right graphs $G_{L}$, $G_R$ defined there.
The structural key for the improvements in the $f$-connected case lies in the following lemma, which strengthens~\Cref{lem:cut-in-G-to-left-right-graphs}.
This is essentially~\cite[Lemma 3.7]{NSYArxiv23}; we provide the proof to keep the presentation stand-alone.

\begin{lemma}\label{lem:f-connected-left-right-graph-completness}
    Suppose $G$ is $f$-connected.
    Let $F \subseteq V$ with $|F| = f$.
    Suppose that in $G$, $F$ separates $T$ but does not separate $S$.
    Then either:
    \begin{itemize}
        \item $F \subseteq L \cup S$ and $F$ separates $T \cap V(G_L)$ in $G_L$, or
        \item $F \subseteq R \cup S$ and $F$ separates $T \cap V(G_R)$ in $G_R$.
    \end{itemize}
\end{lemma}
\begin{proof}
    The proof starts exactly as in~\Cref{lem:cut-in-G-to-left-right-graphs}.
    Let $(A,F,B)$ be a $T$-cut in $G$.
    Because $F$ does not separate $S$, we may assume that $B \cap S = \emptyset$ (otherwise $A \cap S = \emptyset$, so we can swap $A$ and $B$).
    Choose some terminal $t^* \in B \cap T$.
    Then $t^* \notin S$, hence $t^* \in L \cup R$.
    In the proof~\Cref{lem:cut-in-G-to-left-right-graphs}, it is shown that $t^* \in L$ implies that $F \cap V(G_L)$ separates $T \cap V(G_L)$ in $G_L$, and symmetrically for $t^* \in R$.
    So, all that remains to show is that $t^* \in L$ implies $F \subseteq L \cup S$ (and the argument for $t^* \in R$ implies $F \subseteq R \cup S$ is symmetric).

    Consider $Z = N(L \cap B)$, i.e., $Z$ is the set of vertices outside $L \cap B$ with some neighbor in $L \cap B$ (with respect to the graph $G$).
    Note that $Z \cap A = Z \cap R = \emptyset$, since there are no edges between $A$ and $B$, nor between $L$ and $R$.
    Thus, $Z$ separates $t^* \in L \cap B$ from every vertex in $A \cup R \neq \emptyset$, so $Z$ is a cut in $G$.
    
    Finally, we claim that $Z = F$; this will end the proof as we already saw that $Z \cap R = \emptyset$.
    Because $Z$ is a cut in $G$, and $F$ is a \emph{minimum} cut in $G$, it suffices to prove that $Z \subseteq F$.
    Let $u \in Z$.
    Since $Z \cap A = \emptyset$, we get $u \in F$ or $u \in B$; we show that the latter is impossible.
    Seeking contradiction, suppose $u \in B$.
    Then, since $Z = N(L \cap B)$, we have that $u \notin L$.
    Also, because $B \cap S = \emptyset$, we get $u \notin S$.
    Thus, it must be that $u \in R$.
    But, since $Z = N(L \cap B)$, $u$ has some neighbor in $L \cap B \subseteq L$, i.e., there is an edge between $L$ and $R$ --- contradiction.
\end{proof}

Additionally, we observe that $f$-connectivity is inherited from $G$ by $G_L$ and $G_R$, as follows immediately from~\Cref{lem:cut-in-left-right-is-cut-in-G}.
\begin{observation}\label{obs:left-right-f-connected}
    If $G$ is $f$-connected, then so are $G_L$ and $G_R$.
\end{observation}

We now turn to describe the modifications to improve the cut detector $\D$ of~\Cref{thm:cut-detectors} constructed in case $G$ is $f$-connected.

\paragraph{Construction of $\T$ and $S^*$.}
The construction of the $f$-LR tree $\T$ and the new terminal set $S^*$ stays exactly the same as in~\Cref{sec:left-right-tree}.
Note that now, by (inductively applying) \Cref{obs:left-right-f-connected}, every graph $G_q$ associated with a node $q$ of $\T$ is $f$-connected.

\paragraph{Construction of $\D$.}
The only modification we make to $\D$ is in the cut detectors with which we augment the internal nodes of $\T$.
Let $q$ be an internal node in $\T$, associated with the graph-terminals pair $(G_q, T_q)$ and the cut $(L_q,S_q,R_q)$ in $G_q$.
Recall that previously, in~\Cref{sec:construction-and-query-of-D}, we augmented $q$ with two $\USOracle$ structures, one for $G_{L_q}$ and one for $G_{R_q}$, with ``S-set'' $S_q$ and ``U-set'' $U_{L_q} \cup U_{R_q}$.
In the $f$-connected case, this simplifies: we will only need one $\USOracle$ constructed \emph{for $G_q$ itself}, with the same ``S-set'', but with an \emph{empty ``U-set''}.
Namely, we augment the internal node $q$ with $\USOracle(G_q, \emptyset, S_q, f)$.
Other than this modification, the construction of $\D$ is exactly as in~\Cref{sec:construction-and-query-of-D}.

\paragraph{Space.}
As in~\Cref{sec:construction-and-query-of-D}, 
the space required for a specific node $q$ in $\T$ is dominated by the cut detector with which $q$ is augmented, which is either $\FEOracle$, $\TEOracle$ or $\USOracle$s.
Now we only have $\USOracle$s with empty ``U-set'', so its space improves to $O(f |V(G_q)|)$.
The space of $\FEOracle$ or of $\TEOracle$ is $\tilde{O}(f|V(G_q)|)$, by~\Cref{lem:terminal-data-structure} and~\Cref{thm:terminal-expander-DS}.
So, summing over all nodes $q$ in $\T$ and using~\Cref{lem:fLR-Tree}(\ref{prop:linear-space}), we now get total space of $\tilde{O}(f n)$ for $\D$.

\paragraph{Query Algorithm for $\D$.}
The implementation of the query is recursive, in the same manner described in~\Cref{sec:construction-and-query-of-D}.
We only consider queries $F \subseteq V$ with $|F| = f$ (as if $|F| < f$, $F$ cannot be a cut in the $f$-connected $G$, so we can just return ``fail'').
Now, the implementation of a recursive call invoked at node $q$ of $\T$ simplifies to the following:
\begin{mdframed}[style=MyFrame]
\begin{description}
    \item[\textbf{(``Leaf'')}] If $q$ is a leaf:
    Then we query the $(f,T_q,\emptyset)$-cut detector of $G_q$ which is found in $q$
    (which is either $\FEOracle(G_q,T_q,f)$ or $\TEOracle(G_q,T_q,\phi,f)$) with $F_q$,
    and return the answer obtained from this query.
    
    \item[\textbf{(``Trim'')}] Else, if $F_q \subseteq S_q$, then we query $\USOracle(G_q, \emptyset, S_q, f)$ with $F_q$, and return the answer obtained from this query.
    
    \item[\textbf{(``Fail'')}] Else, if $F_q \cap L_q \neq \emptyset$ and $F_q \cap R_q \neq \emptyset$: return \emph{``fail''}.

    \item[\textbf{(``Recurse'')}] Else: we have that either $F_q \subseteq L_q \cup S_q$ or $F_q \subseteq R_q \cup S_q$, but not both.
    \begin{itemize}
        \item If the first option occurs, we recurse only on $q_l$.
        \item If the second option occurs, we recurse only on $q_r$ and on $q_s$ (if it exists).
    \end{itemize}
    If at least one recursive call returned \emph{``cut''}, we return \emph{``cut''}; otherwise we return \emph{``fail''}.
\end{description}
\end{mdframed}
Again, $\T(F)$ denotes the query tree of $F$, induced by $\T$ on the nodes visited during the query.

\paragraph{Correctness.}
We show that~\Cref{lem:corretness-induction} still holds after the modifications in case every $G_q$ is $f$-connected, so the correctness follows exactly as in~\Cref{sec:construction-and-query-of-D}.
\begin{proof}[Proof of~\Cref{lem:corretness-induction} in the $f$-connected case.]
    As before, the proof is by induction from the leaves upwards on $\T(F)$.
    The ``Leaf'' case is sound and complete exactly as in the original proof of~\Cref{lem:corretness-induction}.
    
    In the other cases, the argument for the soundness property is essentially identical to the original proof (only now, in the ``Trim'' case, the soundness simply follows from the fact that $\USOracle(G_q, \emptyset, S, f)$ is an $(f, V(G_q), S_q)$-cut detector for $G_q$).
    
    We now show the completeness property, so assume $F_q$ separates $T_q$ but does not separate $S^* \cap V(G_q)$ (and in particular, does not separate $S_q$) in $G_q$.
    Then, in the ``Trim'' case, we must return \emph{``cut''} by the completeness of the $(f,V(G_q),S_q)$-cut detector $\USOracle(G_q, \emptyset, S, f)$ (from~\Cref{lem:US-data-structure}).
    The ``Fail'' case cannot occur by~\Cref{lem:f-connected-left-right-graph-completness}.
    It remains to consider the ``Recurse'' case.
    By~\Cref{lem:f-connected-left-right-graph-completness}, we either have (i) $F_q \subseteq L_q \cup S_q$ and $F_q$ separates $T_q \cap V(G_{L_q})$ in $G_{L_q}$, or (ii) $F_q \subseteq R_q \cup S_q$ and $F_q$ separates $T_q \cap V(G_{R_q})$ in $G_{R_q}$.
    We consider both cases:
    \begin{itemize}
        \item[(i)] The contrapositive of~\Cref{lem:cut-in-left-right-is-cut-in-G} implies that $F_q$ does not separate $S^* \cap V(G_{L_q})$ in $G_{L_q}$.
        Hence, by induction hypothesis, the recursive call invoked on $q_l$ must return \emph{``cut''}, so we also return \emph{``cut''}.
        
        \item[(ii)] Symmetrically to (i), $F_q$ does not separate $S^* \cap V(G_{R_q})$ in $G_{R_q}$.
        If the stepchild $q_s$ doesn't exist, then we return \emph{``cut''} symmetrically to (i), so assume $q_s$ exists.
        If $F_q$ separates $T_q \cap R_q$ in $G_{R_q}$, then the call on $q_r$ must return \emph{``cut''} by induction hypothesis.
        Otherwise, by~\Cref{lem:stepchild-lemma}, $F_q$ separates $U_{L_q} \cup U_{S_q} \cup U_{R_q}$ in $G_{R_q}$, so the call on $q_s$ must return \emph{``cut''} by induction hypothesis.
        In any case, this means we return \emph{``cut''}.
    \end{itemize}
    The proof is concluded.
\end{proof}

\paragraph{Query Time.}
Now, if we ignore the stepchildren in $\T(F)$, we clearly get a path in $\T$. As $\T$ has depth $d = O(\log |T|)$, the total number of nodes visited during the query $F$ is $O(\log |T|)$.
The time spent in each node is analyzed similarly to~\Cref{sec:construction-and-query-of-D}, but now the time to query a $\USOracle$ structure is $\tilde{O}(f)$ instead of $\tilde{O}(f 2^f)$ by~\Cref{lem:US-data-structure}, as every $G_q$ is $f$-connected.
The bottleneck is now at ``leaf'' nodes, where $\tilde{O}(f^2 / \phi)$ time is spent.
So, we get total query time of $\tilde{O}(f^2/\phi \cdot \log |T|) = \tilde{O}(f^2)$.

\paragraph{Preprocessing Time.}
This is analyzed similarly as in~\Cref{sec:construction-and-query-of-D}, only now we construct US cut detector with an empty ``U-set'', which eliminates the $2^{2f}$ factors from the analysis, and from the running time stated in~\Cref{thm:cut-detectors} for general (not $f$-connected) graphs.

\medskip
This concludes the proof of the improvements to~\Cref{thm:cut-detectors} in case $G$ is $f$-connected, which yields~\Cref{thm:cut-theorem-fconnected}.

%% file: space-improvement.tex
\subsection{Optimizing Space and Preprocessing Time in~\Cref{thm:cut-detectors}}\label{sec:space-improvment}

In this section, we explain how to shave the $O(2^{2f})$ factors in the space and preprocessing time that appeared in~\Cref{sec:construction-and-query-of-D}
(while paying only $\poly(f, \log n)$ factors, which are absorbed in the $\tilde{O}(\cdot)$ notation as we are interested in $f= O(\log n)$), thus concluding the proof of~\Cref{thm:cut-detectors}.
Both of them appear for the same reason: our use of US cut detectors from~\Cref{lem:US-data-structure} with a ``U-set'' of size roughly $2f$.
To eliminate them, we will explain how to modify our oracle such that we only apply US cut detectors with empty U-sets.

The reason we used nonempty U-sets came from the construction of left and right graphs in~\Cref{sec:left-right-graphs}, and the possibility that the query $F$ might intersect the representative sets for the sides, $U_L$ and $U_R$.
These representative sets consist of terminals from $T$ (except in minor corner cases).
If somehow we were promised that $F$ does not contain any terminals, i.e., that $F \cap T = \emptyset$, then such intersections are impossible and we would not need to use the U-sets in the US cut detectors.
This intuition is realized using \emph{hit and miss hashing} tools of Karthik and Parter~\cite{KarthikP21}, which immediately yield the following lemma:

\begin{lemma}[\protect{Direct Application of~\cite[Theorem 3.1]{KarthikP21}}]
    Given $T \subseteq V$ and integer  $f \geq 1$, there is a deterministic algorithm that outputs a family of $k = O((f \log n)^3)$ subsets $T_1, \dots, T_k \subseteq T$ with the following property: For every $F \subseteq V$ with $|F| \leq f$ and $u,v \in T-F$, there exists $1 \leq i \leq k$ such that $T_i \cap F = \emptyset$ (``$T_i$ \emph{misses} F'') and $u,v \in T_i$ (``$T_i$ \emph{hits} $u$ and $v$'').
    The running time is $\tilde{O}(n k)$.
\end{lemma}

Note that $k = \poly(f, \log n) = \tilde{O}(1)$, so $k$ factors can be absorbed in $\tilde{O}(\cdot)$ notations.
We will hinge on the following immediate corollary of the above lemma:
\begin{corollary}\label{cor:hit-miss}
    If $F \subseteq V$ separates $T$ in $G$, then there is some $T_i$ in the family such that $T_i \cap F = \emptyset$ and $F$ separates $T_i$ in $G$.
\end{corollary}

So, to realize the cut detector for $T$ in~\Cref{thm:cut-detectors}, we will actually construct $k$ cut detectors, one for each $T_i$ in the family, in a similar fashion to our original construction of~\Cref{thm:cut-detectors}, only with empty U-sets and some minor changes in the construction.
\begin{lemma}\label{lem:hit-miss-cut-detectors}
    For each $T_i$, we can deterministically compute:
    \begin{itemize}
        \item a set $S_i^* \subseteq V$ such that $|S_i^*| \leq \frac{|T|}{2k}$, and
        \item an $(f,T_i,S_i^*)$-cut detector $\D_i$ restricted to queries $F$ such that $F \cap T_i = \emptyset$, with space $\tilde{O}(n)$ and query time $\tilde{O}(2^f)$. 
    \end{itemize}
    Assuming $G$ has $O(fn)$ edges, the running time can be made $\tilde{O}(n) + \tilde{O}(n^{1+\delta})$ for any constant $\delta > 0$.
    The last term can be made $n^{1+o(1)}$ by increasing the space and query time by $n^{o(1)}$ factors.
\end{lemma}

Before explaining the slight modifications of the original construction to obtain~\Cref{lem:hit-miss-cut-detectors}, we show how it is used to get~\Cref{thm:cut-detectors}.
First, we define $S^* = S^*_1 \cup \cdots \cup S^*_k$, so $|S^*| \leq \frac{1}{2}|T|$.
Then, the $(f,T,S^*)$-cut detector $\D$ simply consists of $\D_1, \dots, \D_k$.
To answer a query $F \subseteq V$ with $|F| \leq f$, we query each $\D_i$ such that $T_i \cap F = \emptyset$, and return \emph{``cut''} if one of these queries returned \emph{``cut''}, or \emph{``fail''} if all of them returned \emph{``fail''}.
The correctness follows from~\Cref{cor:hit-miss}.

To conclude this section, we provide the detailed changes to the original construction for~\Cref{thm:cut-detectors} to obtain its modified version in~\Cref{lem:hit-miss-cut-detectors}:
\begin{itemize}
    \item First, we need minor alterations in~\Cref{sec:left-right-graphs} concerning left and right graphs.
    In~\Cref{def:left-right-graphs}, we simplify $U_R$ to just contain one arbitrary terminal from $R \cap T$ and change $U_L$ in the same fashion.
    By adding our current assumption that $F \cap T = \emptyset$, all of the proofs in this section easily adapt to hold (note that there is no point in having $f+1$ terminals in $U_L$ or $U_R$, because now $F$ cannot intersect these sets).

    \item Next, we describe the modifications to the parameters in the construction of the LR-tree (\Cref{sec:left-right-tree}) so as to obtain $|S^*| \leq \frac{1}{2k} |T|$.
    We reduce $\epsilon$ in~\Cref{eq:eps-balanced-or-exp-lemma} to $\epsilon :=  (c k \log |T|)^{-1}$ for a large enough constant $c > 1$.
    Since $k = \tilde{O}(1)$, the expansion $\phi$ in~\Cref{eq:expansion} remains $\phi = 1/ \poly \log n$.
    The analysis yielding~\Cref{lem:fLR-Tree} remains true, except we improve the bound for $|S^*|$.
    First, we derive that $|S^*| \leq \frac{1}{3}((1+3\epsilon)^{d+1} - 1)|T|$ exactly as before.
    Observe that $(1+3\epsilon)^{d+1} \leq e^{3\epsilon(d+1)} \leq 1 + 6\epsilon (d+1)$ (the first inequality is $1+x \leq e^x$, and the second is $e^z \leq 1+2z$ for $z \in [0,1]$).
    So we get $|S^*| \leq 2\epsilon (d+1) |T| \leq \frac{1}{2k}|T|$ (the last inequality is since $d = O(\log |T|)$, so we can set the constant $c$ large enough to make $2\epsilon (d+1) \leq \frac{1}{2k}$).

    \item Finally, we modify the construction of the $(f,T,S^*)$-cut detector $\D$ in~\Cref{sec:construction-and-query-of-D}, by changing the U-sets in the US cut detectors to be $\emptyset$ instead of $U_{L_q} \cup U_{R_q}$.
    This eliminates the $2^{2f}$ factors in space and preprocessing time, as they came from using the original U-sets.
    This change does not hinder any use of US cut detectors, since now $U_{L_q} \cup U_{R_q} \subseteq T$ but $F \cap T = \emptyset$, so whenever we use a US cut detector the query is contained in the S-set.
\end{itemize}

%% file: cut-respecting-family.tex
\section{Cut Respecting Terminal Expander Decomposition}\label{sec:cut-respecting-family}

In this section, we give a clean graph-theoretical characterization that follows as a corollary from the construction of $f$-LR decomposition trees in~\Cref{sec:left-right-tree}.
Informally, it says that any $n$-vertex graph $G$ can be ``decomposed'' into a collection of terminal expanders (or graphs with very few terminals), such that the terminal cuts in them jointly capture all vertex cuts in $G$ of size $f$ or less.
Further, the total size of all these graphs is rather small, roughly proportional to $fn$.
We call this collection an \emph{$f$-cut respecting terminal-expander decomposition} ($f$-cut respecting TED for short), formally defined, as follows:

\begin{definition}[$f$-Cut Respecting $\phi$-TED]\label{def:f-cut-respecting-TED}
    Let $G = (V,E)$ be an $n$-vertex graph, let $f \geq 1$ be an integer, and let $0 < \phi \leq 1$.
    An \emph{$f$-cut respecting $\phi$-terminal expander decomposition}  for $G$ is a collection $\mathcal{G}=\{(G_1,T_1),\ldots, (G_\ell,T_\ell)\}$ that satisfies the following:
    \begin{itemize}
        \item \emph{(Terminal Expanders or Few Terminals)}
        For every $(G_i,T_i)\in \mathcal{G}$, it holds that $T_i \subseteq V(G_i) \subseteq V$, and $G_i$ is either a $(T_i,\phi)$-terminal expander, or $|T_i|=\tilde{O}(f)$.
    
        \item \emph{(Soundness)}
        For every $G_i$, every vertex cut $F$ in $G_i$ with $|F| \leq f$ is also a cut in $G$.

        \item \emph{(Completeness)}
        For every vertex cut $F$ in $G$ with $|F| \leq f$ or less, there exists some $G_i$ such that $F \cap V(G_i)$ separates the terminals $T_i$ in $G_i$.
                
        \item \emph{(Lightness)}
        $\sum_{i} |V(G_i)|=\widetilde{O}(n)$ and $\sum_{i}|E(G_i)|=\widetilde{O}(fn)$.
    \end{itemize}
\end{definition}

The following theorem essentially follows from the $f$-LR tree construction:
\begin{theorem}
    For any $f \geq 1$, every $n$-vertex graph $G$
    admits an $f$-cut respecting $\phi$-TED with $\phi  = 1 /\poly \log n$ that can be computed deterministically in polynomial time,
    or with smaller $\phi = 1/n^{o(1)}$ but improved running time of $O(m) + fn^{1+o(1)}$.
\end{theorem}

\begin{proof}[Proof sketch.]
    As detailed in~\Cref{sec:left-right-tree}, the computation of the $f$-LR tree for $G$ in fact outputs a pair $(\T, S^*)$ where $\T$ is the tree itself, and $S^*$ is the union of all separators $S_q$ from the vertex cuts in the associated graphs of each node $q$ in $\T$.

    Consider the collection of pairs $(G_q, T_q)$ associated with the leaves of $\T$.
    This collection satisfies the first and last conditions in~\Cref{def:f-cut-respecting-TED}, according to~\Cref{lem:fLR-Tree}, properties~\ref{prop:expander-or-small-leaves} and~\ref{prop:linear-space}.
    (Recall that the expansion $\phi$ is determined by setting the parameter $r$, see~\Cref{eq:eps-balanced-or-exp-lemma,eq:expansion}.
    Letting $r= O(1)$ gives $\phi = 1/\poly \log n$, and $r=\Theta(\log \log n)$ gives $\phi = 1/n^{o(1)}$, which determines the construction time of $(\T, S^*)$ as explained in the preprocessing time analysis in~\Cref{sec:construction-and-query-of-D}.)
    
    Further,
    by inductively using the soundness and completeness properties of $f$-left and $f$-right graphs (\Cref{lem:cut-in-G-to-left-right-graphs} and~\Cref{lem:cut-in-left-right-is-cut-in-G}), one can show that the collection of pairs $(G_q, T_q)$ associated with the leaves of $\T$ fully guarantee the soundness condition of~\Cref{def:f-cut-respecting-TED}, and partially guarantee the completeness condition, i.e., guarantee it only when $F$ separates $T$ but not $S^*$ in $G$.
    (The proof is very similar in essence to the proof of~\Cref{lem:corretness-induction}, only easier as one does not need to do deal with ``trim'' cases.)
    
    This last issue is resolved in a similar manner to the construction of cut oracles from terminal cut detectors 
    explained in~\Cref{sec:overview-oracles}.
    We initialize $T_1 = V$, and continue in iterations $i=1,2,\dots$.
    In iteration $i$, we halt if $T_i \neq \emptyset$, and otherwise apply the $f$-LR tree construction for $(G,T_i)$ which yields a pair $(\T_i, S^*_i)$ with $|S^*_i| \leq \frac{1}{2} |T_i|$ (by~\Cref{lem:fLR-Tree}, Property~\ref{prop:terminal-reduction}), and set $T_{i+1} := S^*$.
    Thus, this process can only continue for $O(\log n)$ iterations.
    We take the collection of all pairs $(G_q, T_q)$ associated in the leaves of \emph{all} of the $O(\log n)$ constructed $f$-LR trees.
    By essentially the same arguments as in~\Cref{sec:overview-oracles} (in the proof of~\Cref{thm:cut-detectors}),
    this collection \emph{fully} gauntness the completeness property of~\Cref{def:f-cut-respecting-TED}.
    The other properties still hold; the lightness and running time only incur an $O(\log n)$ factor.
\end{proof}

%% file: space-lower-bound.tex
\section{Space Lower Bound for Vertex Cut Oracles}\label{sec:spaceLB}

In this section, we show the space lower bound for $f$-vertex cut oracles on $n$-vertex graphs of~\Cref{thm:spaceLB}.
We start with the more interesting case:

\paragraph{The case $2 \leq f \leq n/4$.}
Let $U = \{u_1, \dots, u_{n/2}\}$ and $W = \{w_1, \dots, w_{n/2}\}$.
Consider the following process for generating the edge set $E$ of an $n$-vertex graph with vertex set $V = U \cup W$:
\begin{enumerate}
    \item Add a clique on $W$.
    \item Choose $n/2$ \emph{distinct} subsets $F_1, \dots, F_{n/2} \subseteq W$, each of size $|F_i| = f$.
    Denote the collection of chosen subsets by $\cF = \{F_1, \dots, F_{n/2}\}$
    \item For each $i = 1, \dots, n/2$, add edges between $u_i$ and every vertex in $F_i$.
\end{enumerate}
Note that any graph $G$ generated by this process is $f$-connected: if one deletes less than $f$ vertices from $G$, then every surviving vertex from $U$ must have some surviving neighbor in $W$, and the vertices of $W$ are connected by a clique.

Let $G^0$ and $G^1$ be two graphs generated by the above process, whose corresponding collections $\cF^0 = \{F^0_1, \dots, F^0_{n/2}\}$ and $\cF^1 = \{F^1_1, \dots, F^1_{n/2}\}$ chosen in step 2 are different.
Let $\cO^0$ and $\cO^1$ be $f$-vertex cut oracles built for $G^0$ and $G^1$ respectively.
Because $\cF^0 \neq \cF^1$, there must be some $F \subseteq W$ of size $|F| = f$ which belongs to exactly one of the collections $\cF^0, \cF^1$.
Assume without loss of generality that $F \in \cF^0$ and $F \notin \cF^1$.
We show that on query $F$, the oracle $\cO^0$ returns \emph{``cut''}, while oracle $\cO^1$ returns \emph{``not a cut''}, so these two oracles must be different.
\begin{itemize}
    \item Oracle $\cO^0$: As $F \in \cF^0$, we have $F = F^0_i$ for some $1 \leq i \leq n/2$.
    So in $G^0$, the neighbor set of $u_i$ is $F$, and thus $u_i$ is an isolated vertex in $G^0-F$, meaning $F$ is a cut in $G^0$, so the oracle $\cO^0$ must return \emph{``cut''}.

    \item Oracle $\cO^1$: 
    Consider some $u_i \in U$.
    As $F \notin \cF^1$, we have $F \neq F^1_i$.
    Since $|F| = |F^1_i| = f$, there must be some $w \in F^1_i - F$.
    This $w \in W$ is a neighbor of $u_i$ in $G^1-F$.
    So, we have shown that in $G^1 - F$, every vertex in $U$ has some neighbor in $W$.
    As $G^1[W]$ is a clique, this implies that $G^1 - F$ is connected, so $F$ is not a cut in $G^1$ and the oracle $\cO^1$ must return \emph{``not a cut''}.
\end{itemize}

We are now ready to derive the space lower bound.
Let $\sC$ be the set consisting of all possible choices of $\cF$ in step 2, that is
\[
\sC =
\left\{ \{F_1, \dots, F_{n/2}\} \mid \text{$F_1, \dots, F_{n/2}$ are distinct subsets of $W$, each of size $f$} \right\}
\]
As we just showed, each $\cF \in \sC$ gives rise to a different $f$-vertex cut oracle on some $n$-vertex graph.
Thus, in the worst case, the bit representation of such an oracle must consume $\Omega(\log |\sC| )$ bits.
So, all that remains is to bound $|\sC|$ from below.
We have
\[
|\sC| = \binom{\binom{|W|}{f}}{n/2},
\]
as there are $\binom{|W|}{f}$ different subsets of $W$ of size $f$, and each element in $\sC$ is formed by choosing $n/2$ of those subsets.
Now,
\[
\binom{|W|}{f} \geq \left(\frac{|W|}{f} \right)^f = 2^{f (\log |W| - \log(f))} = 2^{f \log(\frac{n}{2f})}
\]
(the inequality is $\binom{k}{r} \geq \left(\frac{k}{r}\right)^r$ for every $1 \leq r \leq k$).
So, we get
\[
    |\sC| = \binom{\binom{|W|}{f}}{n/2} 
    \geq \left(\frac{\binom{|W|}{f}}{n/2} \right)^{n/2}
    \geq \left(\frac{2^{f \log(\frac{n}{2f})}}{n/2} \right)^{n/2} \\
    = 2^{\frac{n}{2} \left(f \log(\frac{n}{2f}) - \log(\frac{n}{2}) \right)}
\]
(the first inequality is $\binom{k}{r} \geq \left(\frac{k}{r}\right)^r$ once again), and taking logarithms yields
\[
    \log |\sC| \geq \frac{fn}{2} \left(\log\big(\frac{n}{2f}\big) - \frac{1}{f} \log\big(\frac{n}{2}\big) \right).
\]
Recall that $2 \leq f \leq n/4$. We split this range into two intervals and analyze them separately:
\begin{itemize}
    \item If $\big(\frac{n}{2}\big)^{1/10} \leq f \leq \frac{n}{4}$, then $\frac{1}{f} \log\big(\frac{n}{2}\big) = o(1)$ while $\log\big(\frac{n}{2f}\big) \geq 1$, so we get $\log |\sC| = \Omega(fn \log(n/f))$.

    \item If $2 \leq f \leq \big(\frac{n}{2}\big)^{1/10}$, then we have 
    $
    \log\big(\frac{n}{2f}\big) - \frac{1}{f} \log\big(\frac{n}{2}\big)
    \geq
    \left(\frac{9}{10} - \frac{1}{f}\right) \log\big(\frac{n}{2}\big) \geq 0.4 \log\big(\frac{n}{2}\big)
    $,
    so again we get $\log |\sC| = \Omega(fn \log(n/f))$.
\end{itemize}
So anyway, we get a space lower bound of
\[
\Omega( \log |\sC|) = \Omega \Big( fn \cdot \log\big(\frac{n}{f}\big) \Big).
\]

\paragraph{The Case $f=1$.}
This case is much easier.
Consider a path $P = (v_1, \dots, v_n)$.
Suppose we generate a new graph from $P$ by making $\Omega(n)$ binary choices: for every integer $1 \leq k < n/2$ we either add the edge $(v_{2k-1}, v_{2k+1})$ or don't.
Let $\mathcal{G}$ be the family of $2^{\Omega(n)}$ possible graphs generated this way.
Given a $1$-vertex cut oracle constructed for some $G \in \mathcal{G}$, one can recover $G$ itself by querying the oracle with every $v_{2k}$, $1 \leq k < n/2$: the edge $(v_{2k-1},v_{2k+1})$ belongs to $G$ iff the answer for query $v_{2k}$ is \emph{``not a cut''}.
Thus, every graph in the family has a different oracle, so the worst-case space is $\Omega(\log |\mathcal{G}|) = \Omega(n)$.

%% file: acknowledgments.tex
\section*{Acknowledgments}

We thank Yaowei Long, Seth Pettie and Thatchaphol Saranurak for useful discussions, and specifically for pointing us to~\cite{HenzingerL0W17} and its connection to incremental vertex-sensitivity connectivity oracles (discussed in~\Cref{sec:incremental-sensitivity}).
We thank Elad Tzalik and Moni Naor for helpful discussions on the proof of~\Cref{thm:spaceLB}, and Amir Abboud for discussions regarding the online version of OVH.

%% file: lower-bounds.tex
\section{Conditional Lower Bounds for Vertex Cut Oracles}\label{sec:condLB}

In this section we discuss conditional lower bounds on $f$-vertex cut oracles; it is largely based on slight adaptations of results and proofs in~\cite{LongS22,LongS22full}, but we give a stand-alone presentation for the sake of organization and completeness.

\smallskip
\noindent \textbf{(i) Logarithmic Number of Faults. } The first lower bound shows that even when $f$ is logarithmic ($f \geq c \log n$ for some large enough constant $c$), query time polynomially better than $O(n)$ would refute the \emph{Strong Exponential Time Hypothesis (SETH)}:

\begin{restatable}[Strong Exponential Time Hypothesis (SETH)]{conjecture}{SETH}\label{con:seth}
        For every $\varepsilon > 0$ there is some $k = k(\varepsilon) \geq 3$ such that $k$-SAT with $N$ variables cannot be solved in $O(2^{(1-\varepsilon)N})$ time.
\end{restatable}

As many SETH-based lower bounds, the lower bound actually relies on the SETH-hardness of the \emph{Orthogonal Vectors} (OV) problem: given two sets of binary $d$-dimensional vectors $A, B \subseteq \{0,1\}^d$ with $|A| = |B| = n$, determine if there exists some $a \in A$ and $b \in B$ such that $\sum_{i=1}^d a_i \cdot b_i = 0$.
Williams~\cite{WILLIAMS2005357} showed that SETH implies the \emph{Orthogonal Vectors Hypothesis (OVH)}:

\begin{conjecture}[Orthogonal Vectors Hypothesis (OVH)]
    For every $\varepsilon > 0$ there is some $c = c(\varepsilon) > 0$ such that OV with dimension $d = c \log n$ cannot be solved in $O(n^{2-\varepsilon})$ time.
\end{conjecture}

We now give the conditional lower bound:
\begin{theorem}[Slight Adaptation of~\protect{\cite[Theorem 8.9]{LongS22}}]\label{thm:OVimplieshardness}
    Assuming OVH, for every $\varepsilon > 0$ there is $c = c(\epsilon) > 0$ such that, for $f = c \log n$, there is no $f$-vertex cut oracle for $n$-vertex graphs with preprocessing time $O(n^{2-\varepsilon})$ and query time $O(n^{1-\varepsilon})$.
\end{theorem}
\begin{proof}
    Fix $\epsilon > 0$, and let $c = c(\varepsilon) > 0$ be the corresponding constant guaranteed by OVH.
    Suppose towards contradiction that there is an $f$-vertex cut oracle for $n$-vertex graphs with preprocessing time $O(n^{2-\varepsilon})$ and query time $O(n^{1-\varepsilon})$.
    We will use this $f$-vertex cut oracle to solve any OV instance $A, B \subseteq \{0,1\}^f$ with $|A| = |B| = n$.
    
    We construct a graph $G = (V,E)$ with $n + f + 1 = O(n)$ vertices, which is defined only by $A$:
    \begin{align*}
        V &= \{x_a \mid a \in A\} \cup \{y_i \mid 1 \leq i \leq f\} \cup \{z\} \\
        E &= \big\{ (x_a, y_i) \mid \text{$a \in A$, $1 \leq i \leq f$ and $a_i = 1$} \big\} \cup \{(y_i, z) \mid 1 \leq i \leq f\}.
    \end{align*}

    For every $b \in B$, we define $F_b = \{y_i \mid b_i = 0\}$.
        We claim that there is
        some $a \in A$ orthogonal to $b$ iff $G-F_b$ is disconnected.
        If $a \in A$ is orthogonal to $b$, then $b_i = 0$ whenever $a_i =1$, meaning that $y_i \in F_b$ whenever $y_i$ is a neighbor of $x_a$, so $x_a$ is isolated in $G-F_b$.
        Conversely, if $G-F_b$ is disconnected, then some $x_a$ must be disconnected from $z$, so each of its neighbors $y_i$ must belong to $F_b$. Thus, $a_i = 1$ could only happen when $b_i = 0$, so $a$ and $b$ are orthogonal.

    In light of the above claim, we can solve the OV instance by constructing an $f$-vertex cut oracle for $G$, and querying it with $F_b$ for every $b \in B$, which takes $O(n^{2-\varepsilon})$ time for oracle construction, and $n \cdot O(n^{1-\varepsilon})$ time for the queries, hence $O(n^{2-\varepsilon})$ time in total --- contradiction to OVH.
    (Note that constructing $G$ itself takes only $O(fn) = O(n \log n)$ time.)
\end{proof}

\smallskip
\noindent \textbf{(ii) Polynomial Number of Faults.} Next, we consider the regime where $f = \Theta(n^{\alpha})$ for some absolute constant $\alpha \in (0,1)$.
For this regime, we state a plausible \emph{online} version of OVH, which can be seen as the natural ``OV analog'' of the popular \emph{OMv conjecture}~\cite{HenzingerKNS15}.
Assuming this online OVH version, in this regime one cannot get query time polynomially smaller than $O(fn)$.

Formally, consider the following \emph{Online Orthogonal Vectors (Online OV)} problem:
Initially, one is given a set $A$ of size $n$, consisting of binary $d$-dimensional vectors, and can preprocess them.
Then, $n$ queries $b^{(1)}, \dots, b^{(n)} \in \{0,1\}^d$ arrive one by one, and for each $b^{(j)}$ it is required to determine if there is some $a \in A$ such that $\sum_{i=1}^d a_i \cdot b^{(j)}_i = 0$ (before the next query arrives).

\begin{conjecture}[Online OVH~\cite{AbboudPersonal}]\label{conj:online-ov}
    Let $\alpha > 0$ be a constant.
    For every constant $\varepsilon > 0$, there is no algorithm that solves Online OV with dimension $d = \Theta(n^{\alpha})$ with $\poly(n)$ preprocessing time and amortized query time $O(n^{1+\alpha - \varepsilon})$.
\end{conjecture}

We note that in the standard offline setting, the only known way to break this bound is via fast matrix multiplication, which (as conjectured by OMv) does not apply to the online setting~\cite{AbboudPersonal}.
An identical construction as in the proof of~\Cref{thm:OVimplieshardness} now yields:

\begin{theorem}
    Assuming~\Cref{conj:online-ov}, for every constant $\alpha \in (0,1)$ and constant $\varepsilon > 0$, letting $f = \Theta(n^\alpha)$, there is no $f$-vertex cut oracles for $n$-vertex graphs where with $\poly(n)$ preprocessing time and $O((fn)^{1-\epsilon})$ query time.
\end{theorem}

\smallskip
\noindent \textbf{(iii) Linear Number of Faults.} Finally, in the regime $f = \Omega(n)$, one can actually show the same lower bound as above assuming the more standard OMv conjecture.
The construction is actually based on the \emph{OuMv conjecture}, which was shown to follow from the OMv conjecture in~\cite{HenzingerKNS15}, so we only define the former here.
In the OuMv problem, one is initially given an integer $n$ and an $n\times n$ Boolean matrix $M$. Then, $\poly(n)$ many queries arrive.
Each query consists of two vectors $u,v \in {0,1}^n$, and asks for $u^T M v$; the output must be produced before the next query is revealed.

\begin{conjecture}[OuMv Conjecture~\cite{HenzingerKNS15}]\label{con:omv}
    For any constant $\epsilon>0$, there is no algorithm that solving OuMv correctly with probability at least $2/3$ such that the preprocessing time is polynomial on $n$ and the amortized query time is $O(n^{2-\epsilon})$.
\end{conjecture}

\begin{theorem}\label{thm:nlowerbound}
    Assuming OMv conjecture, for every $\epsilon>0$, letting $f = \Omega(n)$, there is no $f$-vertex cut oracle for $n$-vertex graphs with preprocessing time polynomial on $n$ and query time $O(n^{2-\epsilon})$.
\end{theorem}

\begin{proof}
    The proof is based on~\cite[Theorem 8.6]{LongS22}. 
    Given the $n \times n$ input matrix $M$,
    we construct a graph $G=(V,E)$ with
    \begin{align*}
        V &= \{a_1, \dots, a_n, b_1, \dots b_n\} \\
        E &= \{(a_i,b_j)\in A\times B\mid M_{i,j}=1\}\cup \{(a_i,a_j)\mid 1 \leq i < j \leq n\}\cup\{(b_i,b_j)\mid 1 \leq i < j \leq n\}
    \end{align*}
    Given query vectors $u,v \in \{0,1\}^n$, we define $F=\{a_i \in A\mid u_i=0\} \cup \{b_j\in B\mid v_j=0\}$ and query the oracle with $F$.
    If the output is ``$F$ is a cut'', then we have $u^T Mv=0$; otherwise we have $u^T Mv=1$. 
    To see this, notice that $u^T Mv=1$ iff there is some $i$ and $j$ such that $M_{i,j} = u_i = v_j = 1$, which is equivalent to saying that there is an edge in $G$ between $\{a_i\in A\mid u_i=1\}$ and $\{b_j \in B\mid v_j=1\}$.
    Since $A$ and $B$ are cliques, this is equivalent to saying that the subgraph induced on $\{a_i\in A\mid u_i=1\}\cup \{b_j \in B\mid v_j=1\}$ is connected, but this subgraph is $G-F$.

    If the oracle has preprocessing time $\poly(n)$ and query time $O(n^{2-\epsilon})$, this algorithm clearly violates the OuMv conjecture.
\end{proof}

\section{Conditional Lower Bounds for Incremental Sensitivity Oracles}\label{sec:incremental-sensitivity}

This section discusses the ``incremental analog'' of $f$-vertex cut oracles.
Intuitively, one is first given a graph $G$ with some ``switched off'' vertices, and should preprocess it into an oracle that upon a query of $f$ vertices that are asked to be ``turned on'' reports if the resulting turned-on graph is connected.
Formally, we define the problem of \emph{$f$-incremental (or decremental) vertex-sensitivity (global) connectivity oracle}  as follows.
\newcommand{\On}{\mathrm{on}}
\newcommand{\Off}{\mathrm{off}}
\begin{definition}\label{def:sensitivity}
    An \emph{$f$-incremental (or decremental) vertex-sensitivity (global) connectivity oracle} is initially given an undirected simple graph $G=(V,E)$, a set of \emph{switched-off} vertices $F\subseteq V$, and after some preprocessing, answers the following queries.
    \begin{itemize}
        \item (Incremental) Given a set of vertices $F_{\On}\subseteq F$ with $|F_{\On}|\le f$, the algorithm answers whether $G[(V-F)\cup F_{\On}]$ is connected or not.
        \item (Decremental) Given a set of vertices $F_{\Off}\subseteq V-F$ with $|F_{\Off}|\le f$, the algorithm answers whether $G[V-F-F_{\Off}]$ is connected or not.
    \end{itemize}
\end{definition}

Notice that an $f$-decremental vertex-sensitivity (global) connectivity oracle on an undirected graph $G$ and switched-off vertices $F$ is equivalent to an $f$-vertex cut oracle on $G[V-F]$.
We use the term \emph{sensitivity} here to better align with the incremental setting (this terminology is taken from~\cite{LongW24}).

Notice that according to \Cref{thm:main-cut-theorem}, we can achieve an almost linear preprocessing time and $\tilde{O}(2^f)$ query time for the decremental setting.
The purpose of this section is to show a strong separation of query time between incremental and decremental settings: we will prove that, under the \emph{Strong Exponential Time Hypothesis (SETH)}, we cannot hope for a $f$-incremental oracle with polynomial preprocessing time and $n^{1-\epsilon}$ query time for any constant $\epsilon$, even when $f$ is a constant.

We first give the exact definition of $k$-SAT for clarity.

\newcommand{\cC}{\mathcal{C}}
\begin{definition}[$k$-SAT]\label{def:ksat}
A Boolean formula is said to be in \emph{$k$-CNF form} if it is expressed as a conjunction of clauses $\bigvee_{C\in\cC}C$, where each clause $C$ is a disjunction of exactly $k$ literals (a literal being a variable or its negation). The \emph{$k$-SAT problem} is the decision problem of determining whether there exists a truth assignment to the variables that makes the entire formula true.
\end{definition}

We restate the definition of SETH \Cref{con:seth} as follows.

\SETH*

We are now ready to prove the lower bound for the incremental setting.
Similar ideas can be found in Section 4 of \cite{HenzingerL0W17}.

\newcommand{\cU}{\mathcal{U}}
\newcommand{\bbN}{\mathbb{N}}
\newcommand{\eps}{\epsilon}
\newcommand{\epsu}{\delta}
\newcommand{\epss}{\eps_{s}}
\begin{theorem}\label{lem:increamentalhardness}
    Under SETH, for every constants $t\in\bbN^+,\eps\in(0,1)$, there is a sufficiently large constant $f$ depending on $t,\eps$ such that no $f$-incremental vertex-sensitivity (global) connectivity oracle on an $n$ vertex undirected graph with $O(n^t)$ preprocessing time and $O(n^{1-\epsilon})$ query time exists.
\end{theorem}
\begin{proof}
    Suppose $\cO$ is such an oracle stated in \Cref{lem:increamentalhardness}. We let
    \[\eps_t:=\frac{\eps}{4t}\qquad\qquad k=k\left(\eps_t\right)\]
    where $k$ is the function stated in \Cref{con:seth}. It will be clear from the analysis why we set $k$ in this way. 
    
    We will derive an algorithm for $k$-SAT with $N$ variables using $\cO$ in $O(2^{(1-\eps_t)N})$ time, which violates \Cref{con:seth}. 

    We will use the following sparsification lemma. It is a useful tool for solving SAT as it shows that any formula can be reduced to a small amount of formula with linear (on $N$) number of clauses. 

    \begin{lemma}[Sparsification Lemma, \cite{ImpagliazzoPZ01}]\label{lem:sparsification}
        For any $\epss > 0$ and $k \in \mathbb{N}$, there exists a constant  $c = c(\epss, k)$ 
        such that any $k$-SAT formula $F$ with $N$ variables can be expressed as $\Phi = \bigvee_{i=1}^{\ell} \Phi_i$, 
        where  
        $\ell = O\left(2^{\epss N}\right)$,
        and each $\Phi_i$ is a $k$-SAT formula with at most $cN$ clauses. Moreover, this disjunction can be computed in time 
        $O\left(2^{\epss N} \operatorname{poly}(N)\right)$.
    \end{lemma}

    \paragraph{Algorithm.} 
    We use the sparsification lemma \Cref{lem:sparsification} with 
    $\epss=\frac{\eps}{6t}$
    to get 
    $\ell=O\left(2^{\epss N}\right)$
    formulas, each has $c(\epss,k)\cdot N$ many clauses. It suffices to check if there exists an assignment such that one of these formulas can be satisfied.
    Now, we focus on one of these formulas.
    
    Suppose the variable set is $\cU$ and the formula has the clause set $\cC$ where $|\cC|= c(\epss,k)\cdot N$.
    Let $\cU'\subseteq\cU$ be an arbitrary subset of variables of size $\epsu N$ where
    $\delta:=\frac{1}{2t}$.
    Split the clause set $\cC$ to
    \[f:=\frac{c(\epss,k)}{\delta}\]
    many clause sets, each of size $\delta N$. Denote them as $\cC=\cC_1\cup \cC_2...\cup \cC_f$.
    
    Construct a graph $G$ with vertex set
    \[V=2^{\cU'}\cup\left(\bigcup_{i\in[f]}2^{\cC_i}\right) .\]
    Each vertex in $U \in 2^{\cU'}$ represents a truth assignment to the variables in $\cU'$ (variables in $U$ are set to `true', and the remaining variables $\cU' - U$ are set to `false').
    Each vertex in $2^{\cC_i}$ represents a clause subset of $\cC_i$. 

    The edge set $E$ of $G$ is defined as follows. For every assignment $U\in 2^{\cU'}$ and for every $i\in[f],C\in 2^{\cC_i}$, connect an edge from $U$ to $C$ iff at least one clause in $C$ is \emph{not} satisfied by $U$. Moreover, the edge set $E$ contains edges forming a clique on $\cup_{i\in[f]}2^{\cC_i}$.

    We initialize the oracle $\cO$ on $G$ with the switched off vertices being
    \[F:=\bigcup_{i\in[f]}2^{\cC_i} .\]
    
    For every assignment $U'\subseteq \cU-\cU'$ to the variables outside $\cU'$, let $\cC_{i,U'}$ be the subset of clauses in $\cC_i$ that are \emph{not} satisfied by $U'$.
    Notice that $\cC_{i,U'}$ is a vertex in $2^{\cC_i}$.
    We make a query using $\cO$ with 
    \[F_{\On}:=\{\cC_{i,U'}\mid i\in[f]\}.\]
    
    If one of these queries corresponding to some  $U'\subseteq \cU-\cU'$ returns `not connected', then we return `can satisfy'. Otherwise, we return `cannot satisfy'.

    \paragraph{Correctness.} 
    For the first direction, suppose there exists an assignment $U^*\subseteq \cU$ satisfying one of these sparsified $k$-SATs. Let $U^*_1\subseteq \cU'$ be the assignment restricted to $\cU'$ and $U^*_2\subseteq \cU-\cU'$ be the assignment restricted to $\cU-\cU'$. Consider the query corresponding to $U^*_2$. Remember that $\cC_{i,U^*_2}$ are the clauses in $\cC_i$ that are not satisfied by $U^*_2$. Since $U^*=U^*_1\cup U^*_2$ satisfies all clauses, it must be that $U^*_1$ satisfies all clauses in $\cC_{i,U^*_2}$. According to the definition of the edge set, there is no edge from $U^*_1$ to $\cC_{i,U^*_2}$ for every $i\in[f]$. Thus, $U^*_1$ becomes a singleton vertex, so the graph $G[(V-F)\cup F_{\On}]$ is not connected on the query $F_{\On}=\{\cC_{i,U^*_2}\mid i\in[f]\}$.

    In the other direction, suppose the algorithm returns `can satisfy'.
    It means that for a query that corresponds to some assignment $U'\subseteq \cU-\cU'$, the graph is not connected. Remember that $\bigcup_{i\in[f]}2^{\cC_i}$ is a clique, hence so are the `turned-on' vertices $F_{\On}=\{\cC_{i,U'}\mid i\in[f]\}$. Thus, if every vertex in $2^{\cU'}$ is connected to $F_{\On}$, the graph $G[(V-F)\cup F_{\On}]$ is connected. So there must be a vertex $U\in 2^{\cU'}$ such that $U$ has no edge to $F_{\On}$. According to the definition of edge set, this means that $U$ satisfies all clauses in $\cC_{i,U'}$ for every $i\in[f]$. Moreover, $U'$ satisfies all clauses in $\cC_i-\cC_{i,U'}$ according to the definition. Thus, $U\cup U'$ satisfies all clauses in $\cC_i$ for every $i\in[f]$, so $\cC$ can be satisfied.

    \paragraph{Running time.} For convenience, we use $O^*(\cdot)$ to hide polynomial factors on $N$. 
    The sparsification lemma \Cref{lem:sparsification} takes time $O^*(2^{\epss N})$. Then we get $\ell$ formulas, and we solve each formula independently; the final running time should be multiplied by $\ell=O(2^{\epss N})$.
    For each formula, we construct a graph with number of vertices
    \[n=2^{\epsu N}+f\cdot 2^{\epsu N}=O(2^{\epsu N})\]
    and at most
    and number of edges at most
    \[m\le n^2=O(2^{2\epsu N}).\]
    Notice that each edge can be checked in at most polynomial time on $N$. 
    Then, we initialize $\cO$ on the graph $G$, which takes a preprocessing time of 
    \[n^t\le O(2^{t\epsu N})\le O(2^{N/2}).\]
    Then, for each subset of $\cU-\cU'$ we do a query. The number of queries is $2^{(1-\epsu)N}$, each query takes time $O(n^{1-\eps})=O(2^{\epsu(1-\eps) N})$.
    We thus bound the total running time by
    \[O^*\left(2^{\epss N}+2^{\epss N}\cdot \left(2^{2\epsu N}+\underbrace{2^{N/2}+2^{(1-\epsu)N}\cdot 2^{\epsu(1-\eps) N} }_{\text{data structure }}\right)\right).\]
    Plugging in $\epss=\eps/6t,\epsu=1/2t,\eps_t=\eps/4t$, the running time is 
    \[O^*\left(2^{\left(1-\frac{\eps}{3t}\right)N}\right)\le o(2^{(1-\eps_t)N})\]
    which violates SETH (\Cref{con:seth}).

    Notice that $f$ is defined as 
    \[f=\frac{c(\epss,k)}{\epsu}=\frac{c(\eps/6t,k(\eps/4t))}{1/2t},\]
    which is a constant depending on $\eps,t$ (recall that $c$ and $k$ are the fixed functions in \Cref{con:seth,lem:sparsification}).
\end{proof}